%% file: main.tex
\documentclass[runningheads]{llncs}

\usepackage[margin=1in]{geometry}

\usepackage[T1]{fontenc}
\usepackage{tabularx}
\usepackage{xspace}
\usepackage{booktabs}

\usepackage{tikz, enumitem}

\usepackage{amsmath}

\usepackage{amsthm}
\usepackage{float}
\usepackage{url}
\usepackage[colorlinks=true, linkcolor=blue, urlcolor=blue, citecolor=blue]{hyperref}

\usepackage{caption}
\usetikzlibrary{automata,arrows,calc,positioning}
\usetikzlibrary{positioning,chains,fit,shapes,calc}

\input{macros}

\title{Generalized Capacity Planning for the Hospital-Residents Problem\thanks{A preliminary version of this work appeared in IWOCA 2023~\cite{iwoca}.}}

\author{Haricharan Balasundaram \inst{1} \and Girija Limaye \inst{2}\thanks{This work is partially done when the author was a Ph.D. scholar at IIT Madras.} \and
Meghana Nasre \inst{1} \and Abhinav Raja\inst{1}} 

\institute{Indian Institute of Technology Madras, India\\
\email{haricharan12122003@gmail.com, meghana@cse.iitm.ac.in, abhinav.raja03@gmail.com}
\and FLAME University, Pune, India \\
\email{girija.limaye@flame.edu.in}
}

\authorrunning{Balasundaram et. al.}

\input{macros}
\newcommand{\AAA}{\mathcal{A}}
\newcommand{\BBB}{\mathcal{P}}
\newcommand{\mpref}{\succ}
\newcommand{\lpref}{\prec}

\newcommand{\CCQCC}{{\sf \MINSUMSPC_{c1,c2}}}

\newcommand{\s}{a}
\newcommand{\cc}{p}
\newcommand{\NP}{\sf NP}
\newcommand{\Poly}{\sf P}
\newcommand{\OPT}{{\sf OPT}}

\usetikzlibrary{positioning, quotes}
\let\oldforall\forall
\let\forall\undefined
\DeclareMathOperator{\forall}{\oldforall}

\usepackage{xcolor, multirow}
\usepackage{amsmath, amssymb, amsfonts}
\usepackage{algorithm}
\usepackage[noend]{algpseudocode}
\usepackage{tikz}
\usepackage{array}

\begin{document}

\maketitle
\begin{abstract}

The Hospital Residents setting models important problems like school choice, assignment of undergraduate students to degree programs, among many others. In this setting, fixed quotas are associated with the programs that limit the number of agents that can be assigned to them. Motivated by scenarios where {\em all} agents must be matched, we propose and study a generalized capacity planning problem, which allows cost-controlled flexibility with respect to quotas.

Our setting is an extension of the Hospital Resident setting where programs have the usual quota as well as an associated cost, indicating the cost of matching an agent beyond the initial quotas. We seek to compute a matching that matches all agents and is optimal with respect to preferences, and minimizes either a local or a global objective on cost.

We show that there is a sharp contrast -- minimizing the local objective is polynomial-time solvable, whereas minimizing the global objective is $\NP$-hard. On the positive side, we present approximation algorithms for the global objective in the general case and a particular hard case. We achieve the approximation guarantee for the special hard case via a linear programming based algorithm. We strengthen the $\NP$-hardness by showing a matching lower bound to our algorithmic result.
    
\end{abstract}

\input{1-intro}
\input{2-related_work_and_background}
\input{3-algo}
\input{4-2_cost_version}
\input{5-hardness}
\input{6-concl}

\bibliography{refs}

\input{appendix}

\end{document}

%% file: macros.tex
\def\MINSUMSP{\textsc{MinSum}\xspace}
\def\MINSUMSPC{\textsc{MinSumC}\xspace}
\def\MINSUMSPQ{\textsc{MinSumQ}\xspace}
\def\MINMAXSP{\textsc{MinMax}\xspace}
\def\MINMAXSPQ{\textsc{MinMaxQ}\xspace}
\def\MINMAXSPC{\textsc{MinMaxC}\xspace}

\def\HR{\textsc{hr}\xspace}

\def\Plc{P_{\text{LC}}}
\def\Pe{P_{\text{empty}}}

%% file: 1-intro.tex
\section{Introduction}

\label{sec:introduction}

The problem of computing optimal many-to-one matchings under two-sided preferences is extensively investigated in the literature~\cite{GS62,Roth,real_world_example_3,BCCKP19,AAST03}. This setting is commonly known as the Hospital Residents (\HR) setting. 

In the \HR setting~\cite{GS62} we are given a set of agents (residents), a set of programs (hospitals), and a set of mutually acceptable pairs between them. Each agent and every program has a preference ordering over its mutually acceptable partners. Additionally, every program has a positive integral capacity that denotes the maximum number of agents that can be assigned to the program. The goal is to compute a matching, that is, an assignment between agents and programs that is {\em optimal} with respect to preferences and capacities. This setting models several real-world problems, such as assigning students to schools~\cite{AAST03}, elective courses~\cite{real_world_example_3}, assigning medical interns to hospitals~\cite{Roth}, and assigning undergraduate students to university programs~\cite{BCCKP19}.

In this setting, {\em stability} is a well-accepted notion of optimality. An assignment between agents and programs is stable if no agent-program pair has an incentive to deviate from it~\cite{GS62}. It is known that all stable assignments match the same set of agents~\cite{GS62}. In certain applications of the $\HR$ setting, {\em every} agent must be matched. For instance, in school choice~\cite{AAST03} every child must find a school; while matching sailors to billets in the US Navy~\cite{yang2003two,robards2001applying}, every sailor must be assigned to some billet. In the $\HR$ setting, the {\em rigid} upper-quotas limit the number of agents that can be matched in any matching. 

The problem of {\em capacity expansion} is investigated very recently in~\cite{capvar,capplan,ijcai_AbeKI22,chen}. In the capacity expansion problem, the quotas of programs are {\em augmented} to improve the {\em welfare} of the agents. In another work, Gajulapalli~et~al.~\cite{Vazirani} study a two-round mechanism for school admissions in which the goal of the second round is to accommodate more students by suggesting quota increments to schools.

In our work, we are interested in the capacity augmentation problem to ensure that {\em every} agent is matched in a stable matching of the resulting instance. Our setting is similar to the \HR\ setting except that in addition to the (initial) capacities, every program also specifies a {\em cost} of matching additional agents to it. The capacity of programs can be {\em augmented} by spending an additional cost per augmented seat such that a stable matching in the augmented instance matches every agent.

Two special cases of this setting have been studied recently by Chen~et~al.~\cite{chen} and Nasre and Limaye~\cite{iwoca}. In~\cite{chen}, the authors assume that each program has a unit cost per augmented seat. In~\cite{iwoca}, the authors assume that the cost per augmented seat can be a non-negative integer but the initial capacities are zero. In both these works, the two main problems investigated are as follows: given an instance with agents, programs, preferences on both sides, capacities, and costs, augment the instance so that it admits a stable matching that matches every agent such that one of the following goals is achieved.

\begin{itemize}
    \item the maximum augmentation cost spent at a program is minimum
    \item the total augmentation cost spent across all programs is minimum
\end{itemize}

In the generalized setting considered in this paper, we allow {\em both}, that is, an arbitrary positive integral initial capacities (unlike zero initial capacities in~\cite{iwoca}) and arbitrary non-negative integral costs (unlike unit costs in~\cite{chen}).
We are ready to define our problems formally.

\input{1.1-notation_and_problem}

\input{1.1-examples}

\input{1.2-our_results}

\noindent{\em Outline of the paper. } In Section~\ref{sec:related_work_and_background}  we present a brief literature review. In Section~\ref{sec:algo} we present our algorithmic results for \MINMAXSP and general instances of \MINSUMSP. In Section~\ref{sec:two_cost} we present an algorithm for the restricted case of \MINSUMSPC problem with two distinct costs. Section~\ref{sec:hardness} presents our hardness and inapproximability results. We conclude in Section~\ref{sec:concl}.

%% file: 1.1-notation_and_problem.tex
\subsection{Notation and Problem Definition}

We are given a bipartite graph $G = (\mathcal{A} \cup \mathcal{P}, E)$, where $\mathcal{A}$ denotes the set of \textit{agents} and $\mathcal{P}$ denotes the set of \textit{programs}. An edge $(a, p) \in E$ indicates that $a$ and $p$ form an \textit{acceptable agent-program pair}. For each vertex $v \in \mathcal{A} \cup \mathcal{P}$, we define $\mathcal{N}(v)$ to be the set of vertices adjacent to $v$ (that is, the \textit{neighborhood} of $v$). Each vertex $v \in \mathcal{A} \cup \mathcal{P}$ ranks vertices in $\mathcal{N}(v)$ in a strict order, called the \textit{preference list} of vertex $v$. For any vertex $v \in \mathcal{A} \cup \mathcal{P}$, if $v$ prefers $x$ over $y$ then we denote it by $x \succ_v y$, equivalently, $y \prec_v x$. 
The length of the longest preference list of an agent is denoted by $\ell_a$ and the length of the longest preference list of a program is denoted by $\ell_p$.
Each program $p$ has an initial quota $q(p)$ which is a positive integer. In literature, this instance is referred to as the $\HR$ instance (programs being the hospitals and agents being the residents). 
In our setting, in addition to the quotas, a program $p$ has an associated non-negative integral cost $c(p)$. As long as a matching matches upto $q(p)$ many agents to $p$, no cost is incurred. For each agent matched above the quota of $p$, we incur a cost $c(p)$ for exceeding the quota of $p$.

A many-to-one matching (called matching here onwards) $M \subseteq E$ in an \HR instance is an assignment of agents to programs such that each agent is matched to at most one program and each program $p$ is matched to at most $q(p)$ many agents. Let $M(a)$ denote the program to which the agent $a \in \mathcal{A}$ is matched. We say $M(a) = \bot$ if $a$ is unmatched in $M$. Let $M(p)$ denote the set of agents matched to program $p$. We call a program $p \in \mathcal{P}$ under-subscribed in a matching $M$ if $|M(p)| < q(p)$ and fully-subscribed if $|M(p)| = q(p)$. We assume that any agent prefers to be matched (to some program) rather than being unmatched. 

\begin{definition}[Stable Matching]
A pair $(a, p) \in E \setminus M$ is a blocking pair w.r.t. the matching $M$ if $p \succ_a M(a)$ and $p$ is either under-subscribed in M or there exists at least one agent $a' \in M(p)$ such that $a \succ_p a'$. A matching M is said to be stable if there is no blocking pair w.r.t. M.
\end{definition}

Given an \HR instance, a stable matching is guaranteed to exist and can be computed in linear time by the celebrated Gale and Shapley algorithm~\cite{GS62}. A matching $M$ is $\mathcal{A}$-perfect if every agent is matched in $M$. In this work, our goal is to compute a matching that is stable {\em and} $\mathcal{A}$-perfect. 
There exist  simple \HR instances that do not admit an $\mathcal{A}$-perfect stable matching. An $\HR$ instance may admit multiple stable matchings, however, by the Rural Hospitals' Theorem~\cite{Roth}, it is known that all stable matchings in an $\HR$ instance match the {\em same} set of agents. Thus, one can efficiently determine whether an \HR instance admits an $\mathcal{A}$-perfect stable matching.

In this work, to achieve $\mathcal{A}$-perfectness, we consider quota augmentation for programs. Let $\mathbb{N}$ be the set of non-negative integers. Let $\tilde{q}: \mathcal{P} \rightarrow \mathbb{N}$ be a function that maps a non-negative integer to every program $p \in \mathcal{P}$. The $\tilde{q}(p)$ indicates the amount by which the quota at program $p$ should be increased such that the modified $\HR$ instance admits an  $\mathcal{A}$-perfect stable matching $M$ and for every program $p$ the quota of $p$ is equal to  $q(p) + \tilde{q}(p)$.

A trivial quota augmentation wherein for every program $p$, $\tilde{q}(p)$ is set such that $q(p) + \tilde{q}(p) = |\mathcal{N}(p)|$ always results in an $\mathcal{A}$-perfect stable matching.
To control $\tilde{q}(p)$, we  {use the} cost function $c: \mathcal{P} \rightarrow \mathbb{N}$ over the set of programs. Given a matching $M$, the cost incurred at a program $p$ is $\tilde{q}(p) \cdot c(p)$.  In other words, the initial $q(p)$ many positions of a program  $p$ are {\em free}. 

Given an \HR instance along with costs, our goal is to compute an $\mathcal{A}$-perfect, stable matching that incurs the {\em minimum} cost. We consider two natural notions of minimization over cost:

\begin{itemize}
    \item minimize the total cost incurred, that is, $\sum_{p \in \mathcal{P}} ( \tilde{q}(p) \cdot c(p) )$.
    \item minimize the maximum cost incurred at any program, that is, $\max_{p \in \mathcal{P}} \{ \tilde{q}(p) \cdot c(p) \}$.
\end{itemize}

Based on this, we define the \MINSUMSP and the \MINMAXSP problems as follows:

\textbf{\noindent\MINSUMSP problem:} Given $G = (\mathcal{A} \cup \mathcal{P}, E)$, preferences of agents and programs, quota $q(p)$ for every program $p$, a cost function $c$ (defined earlier), the \MINSUMSP problem asks for a quota augmentation function $\tilde{q}$ and an $\mathcal{A}$-perfect stable matching in the augmented instance such that the total cost of augmentation is minimized. 

\textbf{\noindent\MINMAXSP problem:} Given $G = (\mathcal{A} \cup \mathcal{P}, E)$, preferences of agents and programs, quota $q(p)$ for every program $p$, a cost function $c$ (defined earlier), the \MINMAXSP problem asks for a quota augmentation function $\tilde{q}$ and an $\mathcal{A}$-perfect stable matching in the augmented instance such that the maximum augmentation cost spent at a program is minimized.

\smallskip
\noindent Next, we define two special cases of \MINSUMSP and \MINMAXSP.

\begin{enumerate}
    \item \textbf{Unit cost for augmentation:} When every program $p$ has $c(p) = 1$, we denote the problems as \MINSUMSPQ and \MINMAXSPQ. This is equivalent to the setting investigated by Chen~et~al.~\cite{chen}.  
    \item \textbf{Initial quotas being zero:} When every program $p$ has $q(p) = 0$, we denote the problems as \MINSUMSPC and \MINMAXSPC. This is equivalent to the \textit{cost-controlled quotas} setting investigated by Limaye and Nasre~\cite{iwoca}. 
\end{enumerate}

The notion of stability considered earlier, is defined with respect to input quotas. In a setting where initial quotas may be zero,  we use the following well-studied relaxation of stability.
Recall that if the matching $M$ is not stable then there exists a blocking
pair $(a,p)$ with respect to $M$. The blocking pair may arise due to under-subscription of the program or may arise due to the matching $M$ assigning to $p$ an agent that $p$ prefers lower than $a$.
If we allow blocking pairs arising due to an under-subscribed program $p$,
then we get a relaxation of stability, called {\em envy-freeness}~\cite{WR18}.

\begin{definition}[Envy-Freeness]
Given a matching $M$, an agent $a$ has a justified envy (here onwards, called {\em envy}) towards another agent $a'$ if $(a', p) \in M$, $p \succ_a M(a)$ and $a \succ_p a'$. The pair $(a, a')$ is called an envy-pair w.r.t. $M$. A matching $M$ is envy-free if there is no envy-pair w.r.t. $M$.     
\end{definition}

We observe the following about envy-freeness and stability when the initial quotas of all programs are zero.
Let $H$ be an instance in our setting in which the initial quotas of all programs are zero. 
Let $M$ be an $\mathcal{A}$-perfect matching in the augmented instance $H'$. If $M$ is stable in $H'$, then by definition, $M$ is also envy-free in $H'$. Next, suppose that $M$ is an envy-free matching in $H'$. Let $G$ denote the $\HR$ instance wherein the preferences are borrowed from $H'$ (that is, $H$), and for every program $p$,  we set its quota in $G$ to be equal to $|M(p)|$. Then, $M$ is stable in $G$. This implies that when we start with initial quotas of all programs being zero, envy-freeness and stability are equivalent.

%% file: 1.1-examples.tex
\medskip

\noindent{\bf Example.} Consider an instance shown in Fig.~\ref{fig:ex} with three agents $a_1, a_2, a_3$ and three programs $p_1, p_2, p_3$. The tuple $(q, c)$ preceding a program indicates that the program has initial quota $q$ and the cost $c$ of matching a single agent to it. That is, $q(p_1) = 1$, $q(p_2) = 1$, $q(p_3) = 0$ and $c(p_1)=0$, $c(p_2) = 3$, $c(p_3) = 4$. Consider the two $\mathcal{A}$-perfect stable matchings: $M_1 = \{(a_1, p_2), (a_2, p_2),(a_3, p_2)\}$ and $M_2 = \{(a_1, p_2),(a_2, p_3),(a_3, p_2)\}$. The total augmentation cost spent in $M_1$ and $M_2$ is $6$ and $7$ respectively, whereas the maximum augmentation cost spent in $M_1$ and $M_2$ is $6$ and $4$ respectively.

It is straightforward to verify that $M_1$ is the unique optimal solution for \MINSUMSP, whereas
$M_2$ is the unique optimal solution for \MINMAXSP.
\begin{figure}[!ht]
\begin{minipage}{0.4\textwidth}
	\begin{align*}
		\s_1 &: \cc_2 \mpref \cc_1\\
		\s_2 &: \cc_3 \mpref \cc_2\\
		\s_3 &: \cc_2
	\end{align*}
\end{minipage}\hfill
	\begin{minipage}{0.4\textwidth}
		\begin{align*}
			(1,0)\ \cc_1 &: \s_1\\
			(1,3)\ \cc_2 &: \s_1 \mpref \s_2 \mpref \s_3\\
			(0,4)\ \cc_3 &: \s_2
		\end{align*}
	\end{minipage}\hfill
	\caption{An illustrative example for \MINSUMSP and \MINMAXSP.}
	\label{fig:ex}
\end{figure}

%% file: 1.2-our_results.tex
\subsection{Our results}

We show that the \MINMAXSP problem can be solved in polynomial time whereas the \MINSUMSP problem is $\NP$-hard. Moreover, the \MINSUMSP problem is inapproximable under standard complexity-theoretic assumptions.

\begin{theorem}
\label{thm:minmaxsp_exact_algo}
The \MINMAXSP problem can be solved in $O(m \log m)$ time  where $m = |E|$.
\end{theorem}

We say that the preferences of elements of a particular set, say the agent set, are derived from a master list if there is an ordering of the programs and the preferences of all agents respect this ordering. 

\begin{theorem}
        \label{thm:minsumsp_hardness}
        \MINSUMSPC cannot be approximated within a constant factor unless $\Poly = \NP$. 
        The result holds even when the preferences of agents and programs are derived from a master list and there are only two distinct costs in the instance.
\end{theorem}  

The above theorem immediately implies that the \MINSUMSP problem is constant-factor inapproximable. The constant factor inapproximability of the \MINSUMSP problem is known from the result of Chen et al.~\cite{chen}, however their result does not imply Theorem~\ref{thm:minsumsp_hardness}. 
We further show that under the Unique Games Conjecture (UGC)~\cite{khot2008}, \MINSUMSPC cannot be approximated to within $(\ell_a-\epsilon)$ for any $\epsilon > 0$.

\begin{theorem}
    \label{thm:minsumsp_inapproximability}
    \MINSUMSPC cannot be approximated to within a factor of ($\ell_a-\epsilon$) for any $\epsilon > 0$ under UGC. This holds even when the preferences of agents and programs are derived from a master list and there are only two distinct costs.
\end{theorem}

Theorem~\ref{thm:minsumsp_inapproximability} implies that the \MINSUMSP problem is also ($\ell_a - \epsilon$)-inapproximable ($\epsilon > 0$) under UGC. This gives another lower bound for the \MINSUMSP problem. 

We now state our algorithmic results for \MINSUMSP. We present two approximation algorithms for the general instances of the \MINSUMSP problem.

\begin{theorem}
    \label{thm:|P|-approx-algo}
    \MINSUMSP admits a $|P|$-approximation algorithm which runs in $O(m \log m)$ time.
\end{theorem}

\begin{theorem}
\label{thm:l_p-algo-theorem}
\MINSUMSP admits an $\ell_p$-approximation algorithm which runs in $O(m \cdot \ell_p)$ time, where $\ell_p$ denotes the length of the longest preference list of a program.
\end{theorem}

We present an approximation algorithm with a guarantee of $\ell_a$ when the instance has two distinct
costs, thereby meeting the lower bound presented in Theorem~\ref{thm:minsumsp_inapproximability}.

\begin{theorem}
\label{thm:minsumccq_2_cost_approximation}
\MINSUMSPC admits an $\ell_a$-approximation algorithm when the instance has two distinct costs.
\end{theorem}

Our results are summarized in Table~\ref{tab:results}.

\begin{table}[H]
    \centering
    \begin{tabular}{|p{0.16\textwidth}|p{0.4\textwidth}|p{0.4\textwidth}|} \hline
         $\MINMAXSP$ & \multicolumn{2}{c|}{$\MINSUMSP$} \\ \hline\hline
        \vspace{0.12cm}
            In $\Poly$ (Theorem~\ref{thm:minmaxsp_exact_algo})
        & 
            \begin{itemize}[topsep=0pt]
                \item $\NP$-hardness and constant 
         factor inapproximability (follows from Theorem \ref{thm:minsumsp_hardness})
                \item $(\ell_a-\epsilon)$-inapproximability (for any $\epsilon>0$) under UGC (Theorem~\ref{thm:minsumsp_inapproximability}) 
            \end{itemize}         
        & 
            \begin{itemize}[topsep=0pt]
            \item $|\BBB|$-approximation (Theorem~\ref{thm:|P|-approx-algo})
                \item $\ell_p$-approximation (Theorem~\ref{thm:l_p-algo-theorem})
                \item $\ell_a$-approximation for \MINSUMSPC when two distinct costs (Theorem~\ref{thm:minsumccq_2_cost_approximation})
            \end{itemize}            
        \\         
            \hline
    \end{tabular}
    \caption{Summary of our results}
    \label{tab:results}
\end{table}

%% file: 2-related_work_and_background.tex
\section{Related Work and Background}

\label{sec:related_work_and_background}

Capacity planning and or capacity augmentation has received attention in the recent past since it is a natural and practical approach to circumvent rigid quotas for matching problems. The capacity planning problem with a similar motivation as ours is studied extensively under the two-sided preference setting in ~\cite{Vazirani,capvar,capplan,ijcai_AbeKI22,chen,AFACAN2024277}. In the two-round school choice problem studied by Gajulapalli~et~al.~\cite{Vazirani}, their goal in round-2 is to match {\em all} agents in a particular set. This set is derived from the matching in round-1 and they need to match the agents in an envy-free manner (called stability preserving in their work). This can still leave certain agents unassigned. It can be shown that the $\MINSUMSPC$ setting generalizes the matching problems in round-2. We remark that in~\cite{Vazirani} the authors state that a variant of $\MINSUMSPC$ problem (Problem~33, Section~7) is $\NP$-hard. However, they do not investigate the problem in detail.

In the very recent works~\cite{capvar,capplan,ijcai_AbeKI22} the authors consider the problem of distributing extra seats (beyond the input quotas) limited by a {\em budget}
that leads to the best outcome for agents. Their setting does not involve costs, and importantly, $\mathcal{A}$-perfectness is not guaranteed. Bobbio~et~al.~\cite{capvar} show the $\NP$-hardness of their problem. Bobbio~et~al.~\cite{capplan} and Abe~et~al.~\cite{ijcai_AbeKI22} propose a set of
approaches which include heuristics along with empirical evaluations. In our work, we present algorithms with theoretical guarantees. Chen and Cs{\'{a}}ji~\cite{chen} investigate a variant of the capacity augmentation problem 
mentioned earlier and present hardness, approximation algorithms, and parameterized complexity results.
Chen and Cs{\'{a}}ji~\cite{chen} investigate several variants of $\MINSUMSPQ$ with objectives such as Pareto efficiency and student popularity instead of $\mathcal{A}$-perfectness. They show these variants to be hard, which implies the hardness of the corresponding objectives in the generalized $\MINSUMSP$ setting as well. Since these variants do not require $\mathcal{A}$-perfectness, they are trivial in the cost-controlled quotas setting (as in \cite{iwoca}) - if there are no programs with cost $0$, the solution will be the empty matching. If there are programs with cost $0$, the solution can be obtained by matching every agent to their highest-preferred such program.

Capacity augmentation along with costs has also been considered in the one-sided preference list setting, known as the house allocation problem. Here every agent has a preference ordering over a subset of the programs and programs are indifferent between the its neighbours.
In the one-sided preference list setting various notions of optimality like rank-maximilaity, fairness, popularity and weak dominance are studied. 
Kavitha and Nasre~\cite{KavithaN11PopularVar}  and Kavitha~et~al.~\cite{KNP12} address the capacity augmentation problem for the notion of popularity.   It is known that a popular matching is not guaranteed to exist in the one-sided preference list setting. Therefore, their objective was to optimally increase program quotas to create an instance that admits a popular matching. In their setting, the min-max version turns out to be $\NP$-hard whereas the min-sum problem (without the $\mathcal{A}$-perfectness requirement) is polynomial time solvable.

The cost-based quotas considered in our work are also considered in the one-sided preference setting  by Santhini~et~al.~\cite{Santhini}.
In their work, the authors ask for stronger guarantees on the output matching than $\mathcal{A}$-perfectness. This is expressed in terms of a {\em signature} of a matching which allows encoding requirements about the number of agents matched to a particular rank. They consider the problem of computing a min-cost matching with a desired signature and show that it is efficiently solvable. This results in the one-sided setting are in contrast to the hardness and inapproximability results we show for a similar optimization problem under two-sided preferences.

\medskip
Before we proceed to present our results, we discuss important connections between envy-free matchings and stable matchings in an \HR\ instance.

\input{2.1-EFM_to_stable}

%% file: 2.1-EFM_to_stable.tex
\subsection{Envy-free Matchings to Stable Matchings}

\label{sec:efm_to_stable_matching}
As noted earlier, envy-freeness and stability are not equivalent in the \HR setting.
We begin by noting a useful property about any blocking pair with respect to an envy-free matching.

\begin{lemma}\label{lem:underbp}
If $M$ is an envy-free matching that is not stable, then all the blocking pairs have under-subscribed programs.
\label{lemma:undersubscribed_posts}
\end{lemma}

\begin{proof}
   Suppose for the sake of contradiction that
   there exists a blocking pair $(a, p)$ such that the program $p$
   is fully-subscribed ($|M(p)| = q(p)$).  Since $(a,p)$ is a blocking pair,
   we must have $q(p) > 0$.  Hence, $|M(p)| > 1$ and there exists an agent $a' \in M(p)$, 
   such that $a \succ_p a'$. This implies that $(a, a')$ is an envy pair, a contradiction to
   the fact that $M$ is an envy-free matching.
\end{proof}

Next we show that it is possible to convert an envy-free matching $M$ in an \HR instance to a stable matching $M_s$ in the same \HR instance (with the program quotas unchanged) such that the agents that are matched in $M$ remain matched in $M_s$. Algorithm~\ref{algo:promotion-algorithm} gives a simple procedure for the same.

\begin{algorithm}[H]

\caption{Convert an envy-free matching to a stable matching}
\begin{algorithmic}[1]

\Statex {\bf Input: } 
An envy-free matching $M$ in an \HR instance $G$

\Statex {\bf Output: } A stable matching $M_s$ in $G$ such that agents matched in $M$ remain matched in $M_s$
\State $M_s \gets M$

    \While{$\exists$ blocking pair $(a', p)$ w.r.t. $M_s$} \label{line:bp}
        \State{Let $a$ be the highest-preferred agent by $p$ such that $p \succ_a M_s(a)$} \label{line:selectAgent}
        \State $M_s \gets M_s \setminus \{(a, M_s(a))\} \cup \{ (a, p)\}$ \label{line:promote} 
    \EndWhile

    \State \Return $M_s$

\end{algorithmic}
\label{algo:promotion-algorithm}
\end{algorithm}

It is easy to see that no agent who is matched in $M$ gets unmatched in $M_s$. Next, we prove that in each iteration of the algorithm in line \ref{line:promote}, $M_s$ remains envy-free.

\begin{lemma}
    The matching $M_s$ remains envy-free after every execution of line \ref{line:promote}.
    \label{lem:envy_free_remains}
\end{lemma}

\begin{proof}
    Let $a$ be the agent and $p$ be the program selected in line \ref{line:selectAgent}. Agent $a$ is promoted in line~\ref{line:promote}, so $a$ will not envy any new agents after this step. Since $a$ is the highest-preferred agent that forms a blocking pair with respect to $p$, after the promotion in line~\ref{line:promote}, no new agents will envy $a$ either. To see this, note that all agents $a''$ such that $a'' \succ_p a$ are either already matched to $p$ or matched to some $p'$ such that $p' \succ_{a''} p$. Since the initial matching $M$ is envy-free and no new envy pairs are created after each execution of line \ref{line:promote}, $M_s$ remains envy-free.
\end{proof}

\begin{lemma}
Algorithm~\ref{algo:promotion-algorithm} terminates and outputs a stable matching $M_s$.   
\label{lem:envy_free_M_s}
\end{lemma}

\begin{proof}

The {\bf while} loop of the algorithm runs as long as the matching $M_s$ is not stable. At each iteration, since by Lemma~\ref{lem:envy_free_remains} the matching $M_s$ is envy-free and possibly not yet stable, we are guaranteed to find a blocking pair. By Lemma~\ref{lemma:undersubscribed_posts} such a blocking pair is guaranteed to be of an under-subscription type blocking pair. Furthermore, once a blocking pair is found, the agent $a$ gets promoted.  Since agent preference lists are finite, this procedure must terminate in $O(m)$ iterations. Finally, when we exit the {\bf while} loop, the matching does not admit any blocking pair which implies that $M_s$ is stable. 
\end{proof}

This establishes the following lemma. 

\begin{lemma}
    Given an envy-free matching $M$ in an \HR instance, we can obtain a stable matching $M_s$ in the same \HR instance such that the agents matched in $M$ remain matched in $M_s$.
    \label{lem:efm_to_stable}
\end{lemma}

Lemma~\ref{lem:efm_to_stable} implies that if $M$ is an $\mathcal{A}$-perfect, envy-free matching in an \HR\ instance , then there exists an $\mathcal{A}$-perfect, stable matching $M_s$ in the same instance.

%% file: 3-algo.tex
\section{Algorithmic results}\label{sec:algo}

In this section, we present our algorithmic results for the \MINMAXSP problem (Theorem~\ref{thm:minmaxsp_exact_algo}) and our approximation algorithms for  general instances of the
\MINSUMSP problem (Theorem~\ref{thm:|P|-approx-algo}
and Theorem~\ref{thm:l_p-algo-theorem}).

\input{3.1-exact_algorithm_minmaxsp}

\input{3.2-P_approximation}

\input{3.3-l_p_approximation}

%% file: 3.1-exact_algorithm_minmaxsp.tex
\subsection{Polynomial time algorithm for \MINMAXSP}

\label{subsec:minmaxsp}

In this section, we present a polynomial time algorithm for the \MINMAXSP problem. We begin with some observations. Let $M^*$ be an optimal solution for the \MINMAXSP problem and let $t^* = \max_p \{c(p) \cdot \tilde{q} (p)\}$ be the cost of $M^*$. For an integer $t$, we define $G_t$ to be the \HR instance such that for every program $p \in \mathcal{P}$, its quota is $q(p) + \left \lfloor \frac t{c(p)} \right \rfloor$. 
We observe the following:

\begin{itemize}
    \item For any integer $t \ge t^*$, there exists an $\mathcal{A}$-perfect stable matching in $G_t$. This holds because $M^*$ is an $\mathcal{A}$-perfect stable matching in $G_t$ and for every program $p$, $q(p)$ in $G_t$
    is at least as much as $q(p)$ in $G_{t^*}$. Therefore, $M^*$ is an $\mathcal{A}$-perfect envy-free matching in $G_t$. Combining these facts with Lemma~\ref{lem:efm_to_stable} we get the desired result.    
    \item For any integer $t < t^*$, there does not exist an $\mathcal{A}$-perfect stable matching in $G_t$. 
\end{itemize}

In the following lemma, we show an upper bound on the number of distinct values that $t^*$ can take. Recall that $m = |E|$. 

\begin{lemma}
\label{lem:tbound}
There are at most $m + 1$ many distinct values that $t^*$ possibly takes.
\end{lemma}

\begin{proof}
For any program $p$, $0 \le \tilde{q} (p) \le |\mathcal{N}(p)| - q(p)$. This holds because in any solution ($\mathcal{A}$-perfect stable matching), for any program $p$, the initial quota and the quota augmentation together cannot be more than $|\mathcal{N} (p)|$. 

The cost $t^*$ is either equal to $0$ or equal to $c(p) \cdot k$, for some $k$ and some $p \in \mathcal{P}$ such that  $1 \le k \le |\mathcal{N}(p)| - q(p) \le |\mathcal{N}(p)|$. Thus, for a fixed program $p$, we have at most $|\mathcal{N}(p)|$ many non-zero values, if $c(p) > 0$. Thus, $t^*$ takes at most $\sum_{p \in \mathcal{P}} |\mathcal{N}(p)| = |E| = m$ many distinct non-zero values. Including $0$, the lemma holds.
\end{proof}

We now give a polynomial-time algorithm for the \MINMAXSP problem -- construct a sorted array $\hat c$ such that it contains all the possible values of $t^*$. Then perform a binary search on this array. For each cost $t$ being considered, construct the \HR instance $G_t$ and check whether $G_t$ admits an $\mathcal{A}$-perfect stable matching. If yes, search over values less than or equal to $t$, otherwise search over values strictly greater than $t$.

By Lemma~\ref{lem:tbound}, there are $O(m)$ values in the array $\hat c$ and binary search over these costs takes $O(\log m)$ iterations. In each iteration a stable matching is computed and checked for $\mathcal{A}$-perfectness, thereby spending $O(m)$ time per iteration. This implies that the algorithm runs in $O(m \log{m})$ time.
This proves Theorem~\ref{thm:minmaxsp_exact_algo}.

%% file: 3.2-P_approximation.tex
\subsection{$|P|$-approximation algorithm for \MINSUMSP}

Using the algorithm for the \MINMAXSP problem presented in Section~\ref{subsec:minmaxsp}, we present a $|P|$-approximation algorithm for the \MINSUMSP problem.

\begin{lemma}\label{lem:lemma_cost_bound}
    The optimal solution for the \MINMAXSP problem is a $|\mathcal{P}|$-approximation for the \MINSUMSP problem on the same instance.
\end{lemma} 

\begin{proof}
Given an instance $G$, let $M^*$ be an optimal solution for \MINMAXSP and let its cost be $t^*$. Let the optimal cost of \MINSUMSP on the same instance be $y^*$. Clearly, $y^* \ge t^*$ (otherwise it contradicts the optimality of $t^*$ for \MINMAXSP). 

Since $t^*$ is the maximum augmentation cost spent at any program in $M^*$, the total cost of $M^*$ is at most $|\mathcal{P}| \cdot t^*$. Hence, the total cost of $M^*$ is at most $|\mathcal{P}| \cdot y^*$, giving us the desired approximation.
\end{proof}

This proves Theorem~\ref{thm:|P|-approx-algo}.

%% file: 3.3-l_p_approximation.tex
\subsection{$\ell_p$-approximation algorithm for \MINSUMSP}

In this section, we present an $\ell_p$-approximation algorithm for the \MINSUMSP problem. This algorithm is an extension of the algorithm presented in~\cite{iwoca} for a restricted setting wherein all initial quotas are zero. 

Let $G$ be the instance of the \MINSUMSP problem. Let $M_I$ be any stable matching in $G$. We call $M_I$ the \textit{initial} stable matching. If $M_I$ is $\mathcal{A}$-perfect then we return $M_I$. Otherwise, at least one agent is unmatched in $M_I$. Let $\mathcal{A}_u$ denote the set of agents that are unmatched in $M_I$.  It is well known that all stable matchings of an \HR\ instance match the same set of agents~\cite{Roth}. Hence, the set $\mathcal{A}_u$ is invariant of $M_I$.
For every agent $a$, let $p_a^*$ denote the least-cost program occurring in the preference list of $a$. If there are multiple such programs, let $p_a^*$ be the highest-preferred one among them. Using $M_I$ we construct an $\mathcal{A}$-perfect matching by matching every agent $u \in \mathcal{A}_u$ to $p_u^*$. This gives us an intermediate matching, we call it $M_L$. We observe that $M_L$ is $\mathcal{A}$-perfect but not necessarily stable. Next, we promote agents to ensure that $M_L$ is stable, as described below.

We start with the matching $M = M_L$. We consider a program $p$ and consider agents in the reverse order of the preference list of $p$.
If an agent $a$ envies agent $a' \in M(p)$ then we promote $a$ by unmatching it and matching it to $p$ (see Fig.~\ref{fig:promotion_agents_1}). This process repeats for every program.
The pseudo-code is given in Algorithm~\ref{algo:l_p-algorithm}. 

\input{3.3-promotion_posts_envy}

\input{3.3-l_p_algo}

We proceed to prove the correctness. We begin by showing that the matching $M$ computed by Algorithm~\ref{algo:l_p-algorithm} is an $\mathcal{A}$-perfect stable matching.

\begin{lemma}
The matching $M$ computed by Algorithm~\ref{algo:l_p-algorithm} is an $\mathcal{A}$-perfect stable matching.    
\end{lemma}

\begin{proof}
If the algorithm returns the matching $M_I$ then it is an $\AAA$-perfect stable matching.
We show that the matching $M$ returned at line~\ref{line:retM} is $\AAA$-perfect and stable.

It is clear that $M_L$ is $\AAA$-perfect, by construction (line~\ref{line:ml}). Thus, we start with matching $M$ that is $\AAA$-perfect.
During the execution of the loop at line~\ref{line:promoteloop}, agents are only promoted and no agent becomes unmatched. Thus, $M$ remains $\AAA$-perfect at the termination.

Next, we show that $M$ is stable. Suppose for contradiction that there exists a blocking pair $(a, p)$ with respect to $M$. Then $p \succ_a M(a)$ and there exists an $a' \in M(p)$ such that $a \succ_p a'$. 
Consider the iteration of the {\bf for} loop (in line~\ref{line:promoteloop}) when $p$ was considered.
Agent $a'$ was either already matched to $p$ (before the iteration began) or is assigned to $p$ in this iteration. Note that $a \succ_p a'$ and agents are considered in the reverse order of preference of $p$. Thus, in either case, when $a$ was considered in line~\ref{line:agent}, $M(a') = p$ holds. 

If $M(a) \succ_a p$ or $M(a) = p$ at this point, since agents never get demoted in the algorithm, $M(a) \succ_a p$ or $M(a) = p$ holds at the end of the algorithm. 
Thus, we must have that $M (a) \prec_a p$ at this point. This implies that the algorithm matched $a$ to $p$ in this iteration. Since $a$ could only get promoted during the subsequent iterations, $M(a) = p$ or $M(a) \succ_a p$ at the end of the algorithm. This contradicts the claimed blocking pair.
This proves the lemma.
\end{proof}

Next, we show that $M$ is an $\ell_p$-approximation of \MINSUMSP.
Let $\Pe$ denote the set of programs $p$ such that no agent is matched to $p$ in the matching $M_I$ and $\Plc$ denote the set of programs $p$ such that for at least one agent $u \in \mathcal{A}_u$, $p^*_u = p$. Let $S^c$ denote the complement of $S$ with respect to $\mathcal{P}$.
We define the following four sets of programs. Clearly, these four sets are pairwise disjoint (see Fig.~\ref{fig:types_of_programs}).

\begin{itemize}
    \item $P_1 = \Pe^c \cap \Plc^c$
    \item $P_2 = \Pe^c \cap \Plc$
    \item $P_3 = \Pe \cap \Plc$
    \item $P_4 = \Pe \cap \Plc^c$
    
\end{itemize}

\input{3.3-schematic_type_posts}

Now we proceed to prove the approximation guarantee by observing important properties of the programs in these sets.

\begin{lemma}\label{lem:p1p4}
No agent is promoted and matched to any program in $P_1 \cup P_4$ during the execution of the {\bf for} loop in line~\ref{line:promoteloop}.
\end{lemma}

\begin{proof}

Let $p \in P_1 \cup P_4$. By definition of $P_1$ and $P_4$, $p \in \Plc^c$. This implies that for no agent $a \in A_u$, $p = p_a^*$. Therefore, after the execution of line~\ref{line:ml}, $(a,p) \in M_L$ if and only if $(a,p) \in M_I$. Subsequently, at line~\ref{line:mml}, we have $(a,p) \in M$ if and only if $(a,p) \in M_I$.

For the sake of contradiction, assume that some agent is promoted and matched to $p$. Let $a$ be the first such agent. For this promotion to happen, there must exist agent $a' \in M(p)$ such that $a \succ_p a'$ and $p \succ_a M(a)$ in this iteration. By the choice of $a$, we have $(a',p) \in M_I$. 
Note that $M(a) = M_I(a)$ or $M(a) \succ_a M_I(a)$ since agents can only get promoted during the execution of the loop. 
By stability of $M_I$, we have that either $a \in A_u$ and $a' \succ_p a$ or $M_I(a) \succ_a p$. The former case contradicts that $a \succ_p a'$ whereas the latter case contradicts that $M(a) = M_I(a)$ or $M(a) \succ_a M_I(a)$.
This proves the lemma.
\end{proof}

\begin{lemma}
Matching $M$ computed by Algorithm~\ref{algo:l_p-algorithm} is an $\ell_p$-approximation of the \MINSUMSP problem.
\label{lem:l_p-cost-lemma}
\end{lemma} 

\begin{proof}
Let $c(M)$ denote the total augmentation cost of the matching $M$. During the execution of the {\bf for} loop in line~\ref{line:promoteloop}, the cost is spent only when an agent is promoted and matched to some program $p$. By Lemma~\ref{lem:p1p4}, agents are promoted and matched to some program $p$ such that $p \in P_2 \cup P_3$.

Let $len(p)$ denote the length of the preference list of program $p$. During the execution of the {\bf for} loop in line~\ref{line:promoteloop}, at most $len(p)$ many agents can be promoted and matched to the program $p \in P_2 \cup P_3$. Then $c(M) \le \sum_{p \in P_2 \cup P_3} len(p) \cdot c(p)$.

Programs in $P_2 \cup P_3$ are least-cost programs for some agent $a \in A_u$. Let $\OPT$ denote an optimal solution and $c(\OPT)$ denote the optimal cost. Since $\OPT$ must match all agents in $A_u$, $c(\OPT) \geq \sum_{a \in A_u} c(p_a^*) \geq \sum_{p \in P2 \cup P3} c(p)$. 
Moreover, for every $p$, we have $len(p) \leq \ell_p$.
Thus,

\begin{equation*}
c(M) \le \sum_{p \in P_2 \cup P_3} len(p) \cdot c(p) \leq \ell_p \sum_{p \in P_2 \cup P_3} c(p) \le \ell_p \cdot c(\OPT)
\end{equation*}

\end{proof}

Matchings $M_I$ and $M_L$ can be computed in $O(m)$ time using the Gale-Shapley algorithm. The {\bf for} loop in line \ref{line:promoteloop} takes $O(m \ell_p)$ time. Thus, algorithm \ref{algo:l_p-algorithm} runs in $O(m \ell_p)$ time. 
This proves Theorem~\ref{thm:l_p-algo-theorem}.

%% file: 3.3-promotion_posts_envy.tex
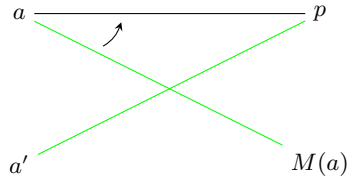
\begin{figure}[!h]
\centering
\begin{tikzpicture}[>=stealth]

\node (a) at (0,0) {$a$};
\node (ap) at (0,-2) {$a'$};
\node (p) at (4,0) {$p$};
\node (pp) at (4,-2) {$M(a)$};
\node (c) at (1, -0.5) {};
\node (d) at (1.4, 0) {};

\draw[-, color=green] (ap) -- (p);
\draw[-, color=green] (a) -- (pp);
\draw[-] (a) -- (p);

\draw[->] (c) to [out = 30, in = -110] (d);
\end{tikzpicture}
\caption{Promotion of agents who envy}
\label{fig:promotion_agents_1}
\end{figure}

%% file: 3.3-l_p_algo.tex
\begin{algorithm}

\caption{$\ell_p$-approximation for \MINSUMSP problem}

\begin{algorithmic}[1]

\Statex {\bf Input:} an instance $G$ of the \MINSUMSP problem

\Statex {\bf Output:} an $\mathcal{A}$-perfect stable matching in an augmented instance of $G$

\State let $M_I$ be a stable matching in the $G$

\If {$M_I$ is $\mathcal{A}$-perfect}
    \State return $M_I$
\EndIf

\State let $M_L$ be the matching obtained as follows:

$M_L \gets \begin{cases}
p_a^* & a \in \mathcal{A}_u \\
M_I(a) & a \notin \mathcal{A}_u
\end{cases}$ \label{line:ml}

\State $M \gets M_L$ \label{line:mml}

    \For{every program $p$} \label{line:promoteloop}
        \For{every agent $a \in \mathcal{A}$ in reverse preference list ordering of $p$}\label{line:agent}
            \If{$\exists a' \in M(p)$ such that $a \succ_p a'$ and $p \succ_a M(a)$}
                \State $M = M \setminus \{(a, M(a)\} \cup \{ (a, p)\}$
            \EndIf
        \EndFor
    \EndFor
    \State \Return $M$\label{line:retM}

\end{algorithmic}
\label{algo:l_p-algorithm}
\end{algorithm}

%% file: 3.3-schematic_type_posts.tex
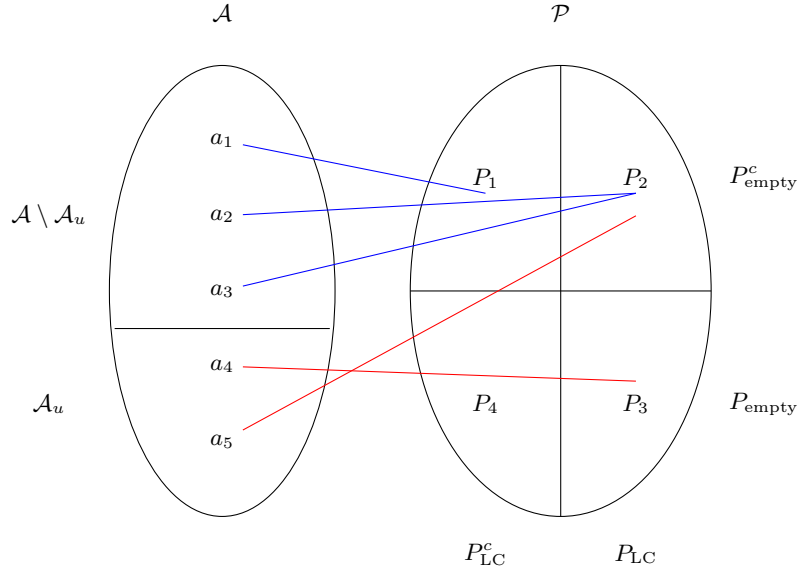
\begin{figure}[!ht]
    \centering

    \begin{tikzpicture}

    \node (f1) at (0, 0) {$a_1$};
    \node (f2) at (0, -1) {$a_2$};
    \node (f3) at (0, -2) {$a_3$};
    \node (f4) at (0, -3) {$a_4$};
    \node (f5) at (0, -4) {$a_5$};

    \node at (-2.3, -1) {$\mathcal{A} \setminus \mathcal{A}_u$};
    \node at (-2.3, -3.5) {$\mathcal{A}_u$};

    \draw (4.5, -2) ellipse (2 and 3);
    \draw (0, -2) ellipse (1.5 and 3);
    \draw (2.5, -2) -- (6.5, -2);
    \draw (4.5, -5) -- (4.5, 1);

    \node at (4.5, 1.7) {$\mathcal{P}$};
    \node at (0, 1.7) {$\mathcal{A}$};

    \node at (5.5, -0.5) {$P_2$};
    \node at (3.5, -0.5) {$P_1$};
    \node at (3.5, -3.5) {$P_4$};
    \node at (5.5, -3.5) {$P_3$};

    \node at (7.2, -0.5) {$\Pe^c$};
    \node at (7.2, -3.5) {$\Pe$};

    \node at (5.5, -5.5) {$\Plc$};
    \node at (3.5, -5.5) {$\Plc^c$};

    \draw[blue] (f1) -- (3.5, -0.7);
    \draw[blue] (f2) -- (5.5, -0.7);
    \draw[blue] (f3) -- (5.5, -0.7);
    \draw[red] (f5) -- (5.5, -1);
    \draw[red] (f4) -- (5.5, -3.2);

    \draw (-1.43, -2.5) -- (1.43, -2.5);

    \end{tikzpicture}

    \caption{A schematic depicting the various types of programs. \textcolor{blue}{Blue} edges denote edges in initial matching $M_I$ while \textcolor{red}{red} edges indicate the edges between an agent and its least-cost program. The graph contains edges from the matching $M_L$.}
    \label{fig:types_of_programs}
\end{figure}

%% file: 4-2_cost_version.tex
\section{\MINSUMSPC with two distinct costs}

\label{sec:two_cost}

In this section, we present a linear program (LP) for the $\MINSUMSPC$ problem followed by an approximation algorithm for a restricted hard case. Recall that under the \MINSUMSPC setting, an envy-free matching is itself stable. Therefore, in this section, we compute an envy-free matching.

\subsection{Linear Program and its dual}

Fig.~\ref{fig:lp-dual} shows the LP relaxation for the $\MINSUMSPC$ problem. Let $H = (\AAA \cup \BBB, E)$ be the underlying graph of the $\MINSUMSPC$ instance. Let $x_{a,p}$ be a primal variable for the edge $(a, p) \in E$: $x_{a, p}$ is $1$ if $a$ is matched to $p$, $0$ otherwise. The objective of the primal LP (Eq.~\ref{eq:lp1}) is to minimize the total cost of all matched edges. Eq.~\ref{eq:lp2} encodes the envy-freeness constraint: if agent $a$  is matched to $p$ then every agent $a' \mpref_p a$ must be matched to either $p$ or a higher-preferred program than $p$, otherwise $a'$ envies $a$. In the primal LP, the envy-freeness constraint is present for a triplet $(a', p, a)$ where $a'\mpref_p a $. We call such a triplet a {\em valid} triplet. Eq.~\ref{eq:lp3} encodes $\AAA$-perfectness constraint.

\begin{figure}[!ht]
{\footnotesize
\setlength\columnsep{55pt}
\noindent{\bf Primal: minimize}
\begin{equation}\label{eq:lp1}
	\sum\limits_{\cc\in \BBB}{c(\cc)\cdot\sum\limits_{(\s,\cc)\in E}{x_{\s,\cc}}}
\end{equation}
\noindent{\bf subject to}
\begin{equation}\label{eq:lp2}
		\sum_{\substack{p': \\p' = p\ \text{or} \\ \cc' \mpref_{\s'} \cc}}{x_{\s',\cc'}} \geq x_{\s,\cc}, \  \forall (\s',\cc)\in E, \s \lpref_{\cc} \s'
\end{equation}
\begin{equation}\label{eq:lp3}
	\sum\limits_{(\s,\cc)\in E}{x_{\s,\cc}} = 1,\ \ \ \forall \s \in \AAA
\end{equation}
\begin{equation}\label{eq:lp4}
	x_{\s,\cc} \geq 0,\ \ \ \forall (\s,\cc) \in E
\end{equation}
\\
\noindent{\bf Dual: maximize}
	\begin{equation}\label{eq:dualobj}
	\sum\limits_{a\in \AAA}{y_a}
\end{equation}
\noindent{\bf subject to}
\begin{equation}
\label{eq:lp5}
		y_a + \sum_{\substack{p': \\p' = p\ \text{or} \\ p' \lpref_{a} p}}\ \ {\sum\limits_{\substack{a':\\a' \lpref_{p'} a}}{z_{a,p',a'}}} - \sum\limits_{\substack{a':\\ a \lpref_p a'}}{z_{a',p,a}} \\\leq c(p),\ \  \ \ \forall (a,p)\in E
\end{equation}
	\begin{equation}\label{eq:lp7}
	z_{a',p,a} \geq 0,\ \ \ \forall (a',p) \in E, a \lpref_p a'
\end{equation}
}
\caption{Linear Program and its dual for the $\MINSUMSPC$ problem}
\label{fig:lp-dual}
\end{figure}

\begin{figure}[!ht]
\centering
 \includegraphics[scale=0.9]{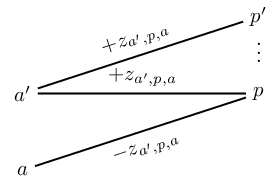}
	\vspace{0.3cm}
	\caption{Let $(a',p,a)$ be a valid triplet and
	$p' \mpref_{a'} p$. The edges shown in the figure are those whose dual constraint contains the variable $z_{a',p,a}$ in either positive or negative form.}
    \label{fig:envy_dual}
\end{figure}

In the dual LP, we have two kinds of variables, the $y$ variables which correspond to every agent and the $z$ variables which correspond to every valid triplet in the primal program. The dual constraint (Eq.~\ref{eq:lp5}) is for every edge $(a, p)$ in $E$. The $y_a$ variable corresponding to an agent $a$ appears in the dual constraint corresponding to every edge incident on $a$. The value $y_a$ can be interpreted as the cost paid by agent $a$ for matching $a$ to one of the programs in $\mathcal{N}(a)$. For an edge $(a, p)$ and an agent $a' \mpref_p a$, the dual variable $z_{a', p, a}$ appears in negative form in exactly one constraint and it is for the edge $(a, p)$. The same dual variable $z_{a',p,a}$ appears in positive form in the constraint for every edge $(a',p')$ such that $p'=p$ or $p' \mpref_{a'} p$ (refer Fig.~\ref{fig:envy_dual}). The value of $z_{a', p, a}$ can be interpreted as the cost paid by agent $a$ in matching $a'$ to a program $p'$ such that $p'=p$ or $p' \mpref_{a'} p$ to resolve potential envy-pair $(a',a)$ if $a$ gets matched to $p$. Following are the useful facts about the linear program and its dual.

\noindent{\bf Fact 1.}
Let $a$ be a fixed agent. If $y_a$ is incremented by a positive value $\Delta$ then it increments the left-hand side (lhs) of the  dual constraint of {\em every} edge $(a,p)$ by $\Delta$ and it does not affect the dual constraint of any edge incident on agent $a' \neq a$.\qed

\noindent{\bf Fact 2.}
Let $(a',p,a)$ be a fixed valid triplet. If $z_{a',p,a}$ is incremented by a positive value $\Delta$ then it increments the lhs of the dual constraint of {\em every} edge $(a',p')$ by $\Delta$ such that $p' = p$ or $p' \mpref_{a'} p$, reduces the lhs of the dual constraint of {\em exactly one} edge $(a,p)$ by $\Delta$ and does not affect the dual constraint of any edge incident on agent $a'' \notin \{a,a'\}$.\qed

For a given dual setting and an edge, if Eq.~\ref{eq:lp5} is satisfied with equality then we call such an edge as a {\em tight edge}, otherwise it is a {\em slack edge}. For an edge $(a,p)$, $slack(a,p)$ denotes its slack. When referring to a $z$ variable, when a specific agent or program occurring in it does not matter, we use $\times$ in its place.

\begin{definition}[Threshold agent] Let $M$ be a matching in the instance. For every program $p$, $thresh(p)$ is the most-preferred agent $a$, if it exists, such that $p \mpref_a M(a)$, otherwise $thresh(p)$ is $\bot$.
\end{definition}

The definition of threshold agent is similar to the threshold resident defined in \cite{MNNR18} and a barrier (vertex) defined in \cite{Vazirani}. We remark that the threshold agent depends on the matching $M$, hence when $M$ gets modified, the threshold agents for programs may change. 

\begin{definition}[Matchable edge]
For an envy-free matching $M$, and an agent $a$ (matched or unmatched), we say that an edge $(a, p) \notin M$ is matchable if the dual constraint on $(a, p)$ is tight and $a = thresh(p)$, otherwise the edge is non-matchable.
\end{definition}

It is straightforward to verify that for an envy-free matching $M$, if we match agent $a$ along a matchable edge then the resultant matching remains envy-free.
		
\subsection{An $\ell_a$-approximation algorithm for $\CCQCC$}\label{sec:two_cost_algo}

\label{sec:lpalgo}

A reader may find the discussion in the Appendix about the challenges involved in designing a primal-dual algorithm for the general case of the \MINSUMSPC\ problem. In this section, we show that when the $\MINSUMSPC$ instance has only two distinct costs $c_1$ and $c_2$ where $c_1 < c_2$, we can circumvent the challenges and obtain an $\ell_a$-approximation algorithm for the $\MINSUMSPC$ problem. We recall from Theorem~\ref{thm:minsumsp_hardness} that even in this restricted setting, the problem remains {\sf NP}-hard.

\noindent {\bf High-level idea of the algorithm. } Our LP based algorithm begins with an initial feasible dual setting and an envy-free matching $M$ which need not be $\AAA$-perfect. As long as $M$ is not $\AAA$-perfect, we pick an unmatched agent $a$ and increase the dual variable $y_a$. We show that for an unmatched agent such an increase is possible and {\em all} edges incident on $a$ become tight due to the update. However, none of the edges incident on $a$ may be matchable (since for every $p \in \mathcal{N}(a)$, $thresh(p) \neq a$). Under the restricted setting of two distinct costs we ensure that after a bounded number of updates to the $z$ variables, at least one edge incident on $a$ is matchable. Throughout we maintain the following invariants with respect to the matching $M$. 

\begin{itemize}
\item $M$ is envy-free, not necessarily $\AAA$-perfect and every matched edge is tight.
\item For an agent $a$ (matched or unmatched), for every $p \mpref_{a} M(a)$, 
	either (i) $(a,p)$
		is tight and $thresh(p) \neq a$ or (ii) $slack(a,p) = c_2-c_1$.
\end{itemize}

Recall that when the matching is modified, thresholds may change,
due to which a tight, non-matchable edge may become matchable.
As long as there exists such an edge, we match it.
This is achieved by the {\em free-promotions} routine. 
The free-promotions routine checks if there exists a matchable edge $(a,p)$. If there is no such edge, the routine terminates.
Otherwise, it matches $(a,p)$, re-computes the threshold agents and repeats the search.
Checking for a matchable edge and computing threshold agents takes
$O(m)$ time where $m = |E|$.
the free-promotions routine runs in $O(m^2)$ time.

\noindent{\bf Description of the algorithm. }
Algorithm~\ref{algo:dualalgo_2} gives the pseudo-code.
In the Appendix, we give an illustrative example 
which depicts the 
key steps of the algorithm on a $\CCQCC$ instance.

We begin with an empty matching $M$ and by setting all $y$ variables to $c_1$ and all $z$ variables to $0$ (line~\ref{line:init1}). Following this, for every
agent $a$ with a cost $c_1$ program in $\mathcal{N}(a)$
we match the agent to its most-preferred program with cost $c_1$ ({\bf for} loop at line~\ref{line:forloop}).
Next, we compute the threshold agent for every program w.r.t. $M$.
As long as $M$ is not $\AAA$-perfect, we pick an arbitrary unmatched agent $a$ and update the dual variables as follows.
	
\begin{algorithm}[!ht]
	\begin{algorithmic}[1]
		\State let $M = \emptyset$, all $y$ variables are set to $c_1$ and all $z$ variables are set to $0$\label{line:init1}
		\For {every agent $a \in \AAA$ s.t. $\exists p \in \mathcal{N}(a)$ such that $c(p) = c_1$}\label{line:forloop}
		    \State let $p$ be the most-preferred program in $\mathcal{N}(a)$ s.t. $c(p) = c_1$ and let
		    $M = M \cup \{(a,p)\}$\label{line:c1p}
		\EndFor
		\State compute $thresh(p)$ for every program $p\in \BBB$\label{line:beforeloop}
		\While {$M$ is not $\AAA$-perfect}\label{line:loop1a2}
			\State let $a$ be an unmatched agent\label{line:picka}
			\While {$a$ is unmatched}\label{line:loop_a}
				\State set $y_a = y_a + c_2 - c_1$ \label{line:yra2}
				\If {there exists a matchable edge incident on $a$}
					\State $M = M \cup \{(a,p) \mid (a,p)$ is the most-preferred matchable edge for  $a\}$\label{line:lineif}
						\State perform free-promotions routine and re-compute thresholds
				\Else
					\State $\BBB(a) = \{ p \in \mathcal{N}(a) \mid p \mpref_a M(a), (a,p)$
						is tight  and $thresh(p) \neq a\}$\label{line:ba}
					\While {$\BBB(a) \neq \emptyset$}\label{line:loop2a2}
						\State let $a'$ be the threshold agent of some program in $\BBB(a)$  \label{line:rprimea2}
						\State let $\BBB(a,a')$ denote the set of programs in $\BBB(a)$ whose threshold agent is $a'$\label{line:paaprime}
						\State let $p$ be the least-preferred program for $a'$ in $\BBB(a, a')$  \label{line:pickp}
						\State set $z_{a',p,a} = c_2-c_1$ \label{line:za2}
						\State let $(a',p')$ be the most-preferred matchable edge incident on $a'$.
						Unmatch $a'$ if matched and
						let $M = M \cup \{(a', p')\}$ \label{line:pro1a2}
	    			    		\State execute free-promotions routine, re-compute thresholds and the set $\BBB(a)$\label{line:bar}
		        		\EndWhile
				\EndIf
			\EndWhile
		\EndWhile
		\State return $M$
	\end{algorithmic}
	\caption{Algorithm to compute an $\ell_a$-approximation of $\MINSUMSPC$ on $\CCQCC$}
	\label{algo:dualalgo_2}
\end{algorithm}

\begin{enumerate}
    \item For the agent $a$, we increase $y_a$ by $c_2-c_1$. We ensure that the dual setting is feasible and all edges incident on $a$ become tight for the dual constraint in Eq.~\ref{eq:lp5}. 
		\label{step:atight}
    Although this step makes all edges incident on $a$ tight, they may not be necessarily matchable. Recall that a tight edge $(a, p)$ is matchable if $thresh(p) = a$. 
    \item If there is a program $p$ such that $(a, p)$ is matchable, then $a$ is immediately matched to the most-preferred such program $p$ (line~\ref{line:lineif})
    and we are done with matching agent $a$. Since the matching is modified, we execute the free-promotions routine.

    \item In case there is no such program for which $a$ is the threshold agent, we update carefully selected $z$ variables in order to either promote the threshold agent (if matched) or match the (unmatched) threshold agent via the following steps.

        \begin{enumerate}
		    \item We compute the set $\BBB(a)$ of programs $p \in \mathcal{N}(a)$
			    such that the dual constraint on edge $(a,p)$ is tight and
			    $thresh(p) \neq a$ and $p \mpref_a M(a)$ (line~\ref{line:ba}).
			    In other words, $\BBB(a)$ is the set of programs in the neighbourhood of $a$
				such that $p$ is higher-preferred over $M(a)$ and edge $(a,p)$ is tight but not matchable.
			    \label{step:ba}
		    \item By the definition of $\BBB(a)$, for every $p_j \in \BBB(a)$,
			    there exists $thresh(p_j) = a' \neq a$.
				We pick an arbitrary agent $a'$ that is a threshold of some program in $\BBB(a)$ (line~\ref{line:rprimea2}).
			    Note that the agent $a'$ can be the threshold agent of more than one program in $\BBB(a)$,
				and we let $\BBB(a,a')$ denote the set of programs in $\BBB(a)$ for whom $a'$ is the threshold.
				Let $p$ be the least-preferred program for $a'$ in $\BBB(a,a')$ (line~\ref{line:pickp}).
				\label{step:pickp}
			\item 	Our goal is to match $a'$ to a program $p'$ such that $p' = p$ or $p' \mpref_{a'} p$.
				By the choice of $a, a'$ and $p$ and from the primal LP, 
				$(a',p, a)$ is a valid triplet and therefore there exists a 
				dual variable $z_{a',p,a}$ (refer Fig.~\ref{fig:envy_dual}).
				We set $z_{a',p,a}$ to $c_2-c_1$ (line~\ref{line:za2}).
				We ensure that this update maintains dual feasibility.
			    \label{step:zupdt}
		    \item 
			    Recall that the variable $z_{a',p,a}$ appears in the positive form in the dual
			    constraint of every edge $(a',p')$ such that $p' = p$ or $p' \mpref_{a'} p$.
		    We ensure that this update results in making all edges $(a',p')$ tight
				and at least one of these 
				becomes matchable.
				We match $a'$ along the most-preferred matchable edge (line~\ref{line:pro1a2}).
				Recall that $z_{a',p, a}$ variable appears in negative form in the dual constraint of edge $(a,p)$,
				hence edge $(a,p)$ becomes slack after this update.
			    \label{step:matchaprime}
		    \item Since $M$ is modified, we execute the free-promotions routine.
			    If a tight edge incident on $a$ becomes matchable, then $a$ is matched inside the free-promotions routine. \label{step:fp}
		    \item We remark that the set $\BBB(a)$ computed in line~\ref{line:ba} is dependent on the matching $M$, specifically $M(a)$ and the threshold agents w.r.t. $M$. In order to maintain a specific slack value on the edges that is useful in maintaining dual feasibility and ensuring progress, we re-compute the set $\BBB(a)$ (line~\ref{line:bar}) and re-enter the loop in line~\ref{line:loop2a2} if $\BBB(a) \neq \emptyset$. \label{step:recompBa}
		\end{enumerate}
\end{enumerate}

\subsection{Proof of correctness}

We start by observing the following properties.

\smallskip
\noindent{\bf (P1)} At line~\ref{line:beforeloop}, no agent is assigned to any program with cost $c_2$ and 
    for every agent $a$ (matched or unmatched), every program $p \mpref_a M(a)$ has cost $c_2$.

\noindent{\bf (P2)}
    A matched agent never gets demoted.

\noindent {\bf (P3)} A tight edge incident on a matched agent 
remains tight.

\noindent {\bf (P4)} All matched edges are tight at the end of the algorithm.

{\bf (P1)} is a simple observation about the matching at line~\ref{line:beforeloop}.
Whenever a matched agent $a$ changes its partner from $M(a)$ to $p'$, 
we have $thresh(p') = a$. By the definition of the threshold agent, $p' \mpref_a M(a)$,
which implies {\bf (P2)}.
Note that the only edge that can become slack during the execution is the edge $(a,p)$ 
which is incident on an unmatched agent $a$ (line~\ref{line:za2}).
This implies {\bf (P3)}.
We observe that when the edge is matched, it is tight.
By {\bf (P3)}, a matched edge (being incident on a matched agent) always remains tight,
implying {\bf (P4)}.

\smallskip

Now, we proceed to prove Theorem~\ref{thm:minsumccq_2_cost_approximation}. In the following lemma, we prove that the matching $M$ computed by the algorithm is envy-free.

\begin{lemma}\label{lem:mefm}
Matching $M$ is envy-free throughout the execution of the algorithm.
\end{lemma}
\begin{proof}
	Matching $M$ is trivially envy-free after line~\ref{line:init1}.
	Any two agents $a$ and $a'$ that are matched in line~\ref{line:c1p}
	are matched to a program with cost $c_1$ and by the choice made in line~\ref{line:c1p},
	it is clear that they do not form an envy-pair.
	By {\bf (P1)}, 
	every unmatched agent $a$ has only cost $c_2$ programs in $\mathcal{N}(a)$
	thus, no unmatched agent envies an agent matched in line~\ref{line:c1p}.
	Thus, $M$ is envy-free before entering the loop at line~\ref{line:loop1a2}.

	Suppose $M$ is envy-free before a modification in $M$ inside the loop.
	We show that it remains envy-free after the modification.
	Matching $M$ is modified either at line~\ref{line:lineif} or line~\ref{line:pro1a2}
	or inside the free-promotions routine. 
	In all these places, only a matchable edge $(a_i,p_j)$ is matched.
	Therefore no agent $a' \neq a_i$ envies $a_i$ after this modification.
	Before this modification $a_i$ did not envy $a' \neq a_i$
	and by {\bf (P2)} $a_i$ (if matched) is not demoted, therefore
	$a_i$ does not envy $a' \neq a_i$ 
	after the modification. Thus, $M$ remains envy-free.
\end{proof}

Next, we proceed to prove the dual feasibility, termination, and $\AAA$-perfectness.
We make the following observation about the innermost {\bf while} loop (line~\ref{line:loop2a2}).

\begin{lemma}\label{lem:boundp}
    Let $a$ be a fixed unmatched agent selected in line~\ref{line:picka}
	and consider an iteration of the loop at line~\ref{line:loop_a} during which
	the algorithm enters {\bf else} part.
	Suppose during an iteration of the loop at line~\ref{line:loop2a2},
	for some $p_k \in \mathcal{N}(a)$, $p = p_k$ is selected at line~\ref{line:pickp}.
	Then at the end of iteration, $slack(a,p_k) = c_2-c_1$ and 
	$p \neq p_k$ during subsequent iterations of the loop.
	Therefore, at most $\ell_a$ many distinct $z_{\times,p_k,a}$ variables are
	updated during the iteration of the loop at line~\ref{line:loop_a}.
\end{lemma}
\begin{proof}
By the choice of $p_k$, 
the edge $(a,p_k)$ was tight before this iteration.
	By {\bf Fact 2}, the update on $z_{\times,p_k,a}$ reduces the lhs of the dual constraint of
the edge $(a,p_k)$ by $c_2-c_1$.
Thus, after this update, $slack(a,p_k) = c_2-c_1$.
Therefore, when $\BBB(a)$ is re-computed at line~\ref{line:bar}, $p_k \notin \BBB(a)$.
	Also observe that no other dual update in $z_{\times,p_j,a}$ inside the loop at line~\ref{line:loop2a2}
	for $p_j \neq p_k$
	affects the slack of edge $(a,p_k)$.
Thus, in a subsequent iteration of this loop, $p_k$ is never selected as $p$ again.

	For every $p_k$ selected as $p$ in line~\ref{line:pickp},
	a distinct $z_{\times,p_k,a}$ variable is updated.
	Thus, there are at most $\mid\hspace{-0.1cm}\BBB(a)\hspace{-0.1cm}\mid$ many distinct $z_{\times,p_k,a}$ variables are
	updated inside the loop at line~\ref{line:loop2a2} in an iteration of the loop at line~\ref{line:loop_a}.
	By observing that $\BBB(a) \subseteq \mathcal{N}(a)$, we get $\mid\hspace{-0.1cm}\BBB(a)\hspace{-0.1cm}\mid \leq \ell_a$, hence the claim follows.
\end{proof}

Recall that if edge $(\hat{a},\hat{p})$ is non-matchable then
either $(\hat{a},\hat{p})$ is slack or $thresh(\hat{p}) \neq \hat{a}$.
In our algorithm, we maintain a stronger invariant: 
for every agent $a$ and for every program $p$ higher-preferred 
over $M(a)$, we maintain that either {\em all} non-matchable edges $(a,p)$ are slack or 
	{\em all} non-matchable edges $(a,p)$ are tight and
for {\em every} such edge,
$thresh(p) \neq a$.
Moreover, we also maintain a specific slack value when the edges are slack.
We categorize agents based on these two cases (see Fig.~\ref{fig:type1} and Fig.~\ref{fig:type2}).
\begin{definition}[type-1 and type-2 agents]
An agent $a$ is called a {\em type-1} agent if for every program $p \mpref_{a} M(a)$,
$slack(a,p) = c_2-c_1$. 
An agent $a$ is called a {\em type-2} agent if $a$ is matched and
for every program $p \mpref_{a} M(a)$,
$slack(a,p) = 0$ and $thresh(p) \neq a$.
\end{definition}

\begin{figure}[!ht]
	\centering
	\includegraphics[scale=0.9]{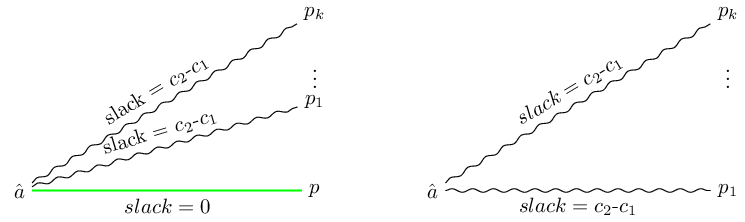}
		\caption{Type-1 agent $\hat{a}$: $slack(\hat{a}, p_j) = c_2 - c_1$, $\forall p_j \mpref_{\hat{a}} p = M(\hat{a})$, if $\hat{a}$ is matched, otherwise, $\forall p_j \in \mathcal{N}(\hat{a})$}
		\label{fig:type1}
	\end{figure}

\begin{figure}[!ht]
	\centering
	\includegraphics[scale=0.9]{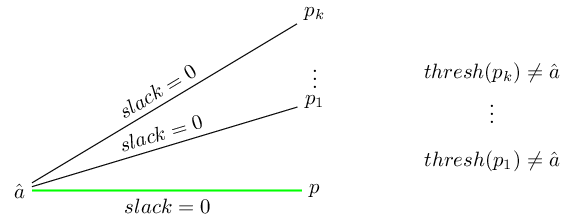}
		\caption{Type-2 agent $\hat{a}$: $slack(\hat{a},p_j) = 0$ $\forall p_j \mpref_{\hat{a}} p = M(\hat{a})$}
        \label{fig:type2}
\end{figure}

We remark that type-1 agents could be either matched or unmatched but type-2 agents are always matched.
Recall that if $a' = a_j$ is unmatched then $M(a_j) = \bot$ and therefore, every program $p_j \in \mathcal{N}(a_j)$
satisfies the condition that $p_j \mpref_{a_j} M(a_j) = \bot$.
We claim that a type-1 agent is selected as $a'$ at most once inside the loop at line~\ref{line:loop2a2}.

\begin{lemma}\label{lem:boundaprime}
	Let $a_j$ be a type-1 agent such that $a' = a_j$ is selected in an arbitrary iteration of the loop
	at line~\ref{line:loop2a2}. Then, at the termination of the loop, $a_j$ is a type-2 agent
	and in subsequent iterations of the loop, $a' \neq a_j$.
\end{lemma}
\begin{proof}
Since $a_j$ is a type-1 agent, for every program $p_j \mpref_{\hat{a}} M(\hat{a})$,
$slack(a_j, p_j) = c_2-c_1$.
Suppose $p = p_k$ is selected in line~\ref{line:pickp}.
Then by {\bf Fact 2}, for every $p_t$ such that $p_t = p_k$ or $p_t \mpref_{a_j} p_k$,
the dual update in line~\ref{line:za2} results in making all $(a_j,p_t)$ edges tight.
Also, since $thresh(p_k) = a_j$, at least one of these newly tight edges (specifically, $(a_j, p_k)$)
becomes matchable.
Therefore, $M(a_j)$ is modified inside the iteration (line~\ref{line:pro1a2}),
implying that $a_j$ is either matched or promoted.
	The choice of $M(a_j)$, that is, $p'$ in line~\ref{line:pro1a2}
	is such that for every $p_j \mpref_{a_j} M(a_j) = p'$, the edge $(a_j,p_j)$ is tight and $thresh(p_j) \neq a_j$.
	Thus, when the iteration ends, $a_j$ is a type-2 agent.

	By {\bf (P3)}, the tight edges incident on $a_j$ remain tight throughout the algorithm.
	In subsequent iterations, agent $a_j$ may further get promoted by the free-promotions routine 
such that for every $p_j \mpref_{a_j} M(a_j)$, $slack(a_j,p_j) = 0$ and $thresh(p_j) \neq a_j$.
Therefore, $a_j$ remains a type-2 agent in all subsequent iterations of the loop.
	This implies that $a_j$ is not the threshold for any program $p_j \mpref_{a_j} M(a_j)$,
	in particular for any program $p_j \in \mathcal{N}(a)$ for the chosen $a$.
Thus, during subsequent iterations of the loop,
$a' \neq a_j$.
\end{proof}

In Lemma~\ref{lem:c1_unm_cond},
we establish that at a specific step during the algorithm, every agent is either type-1 or type-2.
This property is crucial in showing dual feasibility and termination.

\begin{lemma}\label{lem:c1_unm_cond}
	Before every iteration of the loop starting at line~\ref{line:loop_a}, an agent $\hat{a}$ is either
	a type-1 agent or a type-2 agent.
\end{lemma}
\begin{proof}
We prove this by induction.
Before the first iteration of the loop at line~\ref{line:loop_a},
	suppose agent $\hat{a}$ is matched. Then {\bf (P1)} and the initial dual setting together imply that
	for every program $p_j \mpref_{\hat{a}} M(\hat{a})$, $slack(\hat{a},p_j) = c_2-c_1$.
	Therefore $\hat{a}$ is a matched type-1 agent.
	Suppose $\hat{a}$ is unmatched. Then, by {\bf (P1)}, every program $p_j \in \mathcal{N}(\hat{a})$,
	$c(p_j) = c_2$, therefore the initial dual setting implies that $slack(\hat{a}, p_j) = c_2-c_1$.
	This implies that $\hat{a}$ is an unmatched type-1 agent.

	Consider an arbitrary agent $\hat{a}$.
	Suppose that $\hat{a}$ is either type-1 or type-2
	before $l$-th iteration of the loop.
	It is clear that $a$ selected in line~\ref{line:picka} is different than $a'$ selected at line~\ref{line:rprimea2}.
	During the $l$-th iteration, either $a = \hat{a}$ in line~\ref{line:picka}
	or $a' = \hat{a}$ in line~\ref{line:rprimea2} or $\hat{a}$ is promoted inside the free-promotions routine.
	We show that in each of the cases, $\hat{a}$ is either type-1
	or type-2 before $(l+1)$-th iteration begins.

        \begin{enumerate}[label={(\roman*)}]
	\item {\bf $a = \hat{a}$ in line~\ref{line:picka}: }
		It implies that $\hat{a}$ is unmatched.
		By induction hypothesis, $\hat{a}$ is a type-1 agent, therefore
		for every $p_j \in \mathcal{N}(\hat{a})$,
			$slack(\hat{a}, p_j) = c_2-c_1$.
			Then, the update in line~\ref{line:yra2} results in making all edges incident on $\hat{a}$ tight.
			We consider the following two cases -- $\hat{a}$ remains unmatched during
			the $l$-th iteration or $\hat{a}$ gets matched.
			\begin{itemize}
				\item {\bf $\hat{a}$ remains unmatched during the $l$-th iteration: }
			Then the while loop at line~\ref{line:loop2a2} must have been executed.
			During an iteration of the loop at line~\ref{line:loop2a2},
			if $p = p_j$ then the slack of the edge $(\hat{a}, p_j)$ becomes $c_2-c_1$
					after the dual update in line~\ref{line:za2} (by {\bf Fact 2}).
			We show that for every $p_j \in \mathcal{N}(\hat{a})$, 
			there is some iteration of the loop at line~\ref{line:loop2a2}
			such that $p = p_j$ is selected, thereby implying that when the loop terminates,
			for every edge $(\hat{a}, p_j)$, slack becomes $c_2-c_1$.
					Once this is shown, it is clear that before the $(l+1)$-th iteration,
					$\hat{a}$ is a type-1 agent.

			Suppose for contradiction that for some program $p_j$, $p = p_j$ is never selected.
			Since the edge $(\hat{a}, p_j)$ is tight before the loop execution began, it must be the case
			that either $p_j \lpref_{\hat{a}} M(\hat{a})$ or $thresh(p_j) = \hat{a}$.
			The first case implies that $M(\hat{a}) \neq \bot$, a contradiction that $\hat{a}$ remains unmatched during
			the $l$-th iteration. In the second case, since $thresh(p_j) = \hat{a}$, the edge $(\hat{a},p_j)$
			was matchable inside the free-promotions routine, thus $\hat{a}$ must have been matched inside the free-promotions routine, leading to a contradiction again.
			Thus, for every $p_j \in \mathcal{N}(\hat{a})$, there is some iteration of the loop during which
					$p = p_j$. This implies that when the loop at line~\ref{line:loop2a2}
					terminates, for every $p_j \in \mathcal{N}(\hat{a})$,
					$slack(\hat{a}, p_j) = c_2-c_1$.

		\item {\bf $\hat{a}$ gets matched during the $l$-th iteration: }
			Recall that all edges incident on $\hat{a}$ are tight after the dual update in line~\ref{line:yra2}.
			If $\hat{a}$ is matched at line~\ref{line:lineif} then 
			the $l$-th iteration immediately terminates. Thus, before the $(l+1)$-th iteration, 
					for every $p_j \mpref_{\hat{a}} M(\hat{a})$, $slack(\hat{a}, p_j) = 0$
					and by the choice made in line~\ref{line:lineif}, $thresh(p_j) \neq \hat{a}$,
					implying that $\hat{a}$ is a type-2 agent.

			If $\hat{a}$ is matched inside the loop at line~\ref{line:loop2a2} then the free-promotions routine must have matched it.
					Consider the last iteration of the loop at line~\ref{line:loop2a2}
					during which the free-promotions routine matched or promoted $\hat{a}$ 
					and let $M(\hat{a}) = p_t$.
					We will show that for $p_j \mpref_{\hat{a}} p_t$, $slack(a_j,p_j) = c_2-c_1$,
					thereby implying that $\hat{a}$ is a matched type-1 agent
					before $(l+1)$-th iteration begins.

					By Lemma~\ref{lem:boundp}, it is enough to show that for every $p_j \mpref_{\hat{a}} p_t$,
					$p = p_j$ is chosen is some iteration of the loop at line~\ref{line:loop2a2}.
					Suppose not. Then, there exists some $p_j$ such that $(\hat{a}, p_j)$ is tight
					after the loop at line~\ref{line:loop2a2} terminates.
					By the choice of $p_t$ inside the free-promotions routine, 
					$(\hat{a}, p_j)$ was non-matchable, implying that $thresh(p_j) \neq \hat{a}$.
					Hence during the last iteration of the loop, when $\BBB(\hat{a})$ was re-computed
					in line~\ref{line:bar}, $p_j \in \BBB(\hat{a})$, that is, $\BBB(\hat{a}) \neq \emptyset$.
					This contradicts that the loop terminated after this iteration.
					Therefore, for every $p_j \mpref_{\hat{a}} p_t$, $p_j$ was selected in some iteration of the loop at line~\ref{line:loop2a2}, thereby implying that before the $(l+1)$-th iteration of the loop at line~\ref{line:loop_a},
					$\hat{a}$ is a matched type-1 agent.
			\end{itemize}

		\item {\bf $a' = \hat{a}$ at line~\ref{line:rprimea2}}:
		Consider the first iteration of the loop at line~\ref{line:loop2a2} when this happens.
		Note that the dual update in line~\ref{line:yra2} does not affect the slack on edges incident on $\hat{a}$.
		Since $\hat{a}$ is a threshold for some program $p_j \mpref_{\hat{a}} M(\hat{a})$,
		by the induction hypothesis, $\hat{a}$ is a type-1 agent.
		Therefore, for every $p_j \mpref_{\hat{a}} M(\hat{a})$, $slack(\hat{a},p_j) = c_2-c_1$ 
		before this iteration of the loop at line~\ref{line:loop2a2}.
		By Lemma~\ref{lem:boundaprime}, $\hat{a}$ is a type-2 agent when the loop terminates.
		Therefore when $(l+1)$-th iteration of the loop at line~\ref{line:loop_a} begins,
		$\hat{a}$ is a type-2 agent.

	\item {\bf $a \neq \hat{a}$ and $a' \neq \hat{a}$ but $\hat{a}$ is promoted inside the free-promotions routine}:
		First note that none of the dual updates in the $l$-th iteration affect any edge incident on $\hat{a}$.
		Thus, if $\hat{a}$ is promoted inside the free-promotions routine, 
		then by the induction hypothesis, $\hat{a}$ must be a type-2 agent.
		Thus, for every $p_j \mpref_{\hat{a}} M(\hat{a})$, $slack(\hat{a}, p_j) = 0$ and $thresh(p_j) \neq \hat{a}$
		and some update in the matching must have made one of these edges matchable, that is,
		for some tight edge $(\hat{a},p_j)$, $thresh(p_j) = \hat{a}$.
		Consider the last iteration of the loop at line~\ref{line:loop2a2} when the
		free-promotions routine promoted $\hat{a}$.
		Then, by the choice of $M(\hat{a})$ inside the routine, for every program $p_j \mpref_{\hat{a}} M(\hat{a})$,
		edge $(\hat{a}, p_j)$ is non-matchable. This implies that for every such $p_j$,
		$thresh(p_j) \neq \hat{a}$. Thus, $\hat{a}$ remains a type-2 agent
		when the $(l+1)$-th iteration begins.
\end{enumerate}
This completes the proof of the lemma.
\end{proof}

Next, we show that the dual setting is feasible.

\begin{lemma}\label{lem:dualfeas}
	The dual setting is feasible throughout the algorithm.
\end{lemma}
\begin{proof}
It is clear that the dual setting is feasible before entering the loop after line~\ref{line:loop_a}
	for the first time.
	We show that if the dual setting is feasible
	before an arbitrary dual update (either line~\ref{line:yra2} or line~\ref{line:za2})
	then it remains feasible after the update.

\begin{itemize}
	\item {\bf Update at line~\ref{line:yra2}:} 
		Since $a$ is unmatched, by Lemma~\ref{lem:c1_unm_cond},
		$a$ is a type-1 agent and therefore, the slack on every edge 
		incident on $a$ is $c_2-c_1$.
		By {\bf Fact 1}, this update increases the lhs of
		every edge incident on $a$ by $c_2-c_1$ and the iteration
		of the loop at line~\ref{line:loop_a} terminates.
		Therefore the dual setting is feasible.
	\item {\bf Update at line~\ref{line:za2}:}
		We note that the update in line~\ref{line:za2} increases
		the lhs of a subset of edges incident on agent $a'$ (by {\bf Fact 2}).
		Therefore we show that for an arbitrary agent $a_j$ selected
		as $a'$, the dual setting on the affected edges is feasible
		after the update.

		Consider the first iteration of the loop at line~\ref{line:loop2a2}
		wherein an arbitrary $a_j$ is selected as $a'$ in line~\ref{line:rprimea2}.
		Since $a \neq a' = a_j$, the type of $a_j$ before 
		execution of the loop at line~\ref{line:loop2a2} began is same
		as its type before entering the loop at line~\ref{line:loop_a}.
		Suppose $a_j$ is a type-2 agent then the fact that 
		$a_j$ is threshold at some program in $\BBB(a)$ contradicts 
		that for every program $p_j \mpref_{a_j} M(a_j)$, $thresh(p_j) \neq a_j$.
		Therefore, $a_j$ is a type-1 agent.
		This implies that for every $p_j \mpref_{a_j} M(a_j)$, the slack 
		of the edge $(a_j,p_j)$ is $c_2-c_1$, therefore the dual update in
		line~\ref{line:za2} maintains dual feasibility.
		By Lemma~\ref{lem:boundaprime}, this is the only
		iteration of the loop at line~\ref{line:loop2a2}
		when $a' = a_j$.
		Therefore, when the execution of 
		loop at line~\ref{line:loop2a2} terminates (followed by immediate termination of 
		the loop at line~\ref{line:loop_a}), the dual setting remains feasible.
\end{itemize}
This completes the proof of the lemma.
\end{proof}

Now, we show that the algorithm terminates in polynomial time
and computes an $\AAA$-perfect matching $M$.
\begin{lemma}\label{lem:terminates}
	Algorithm~\ref{algo:dualalgo_2} terminates by computing an $\AAA$-perfect matching in polynomial time.
\end{lemma}
\begin{proof}
	We first show that in every iteration of the loop in line~\ref{line:loop_a},
	either an unmatched agent is matched or
	at least one agent is promoted:
	by Lemma~\ref{lem:c1_unm_cond} and {\bf Fact 1},
	after the dual update in line~\ref{line:yra2} all edges incident on $a$ become tight.
	Either $a$ gets matched in line~\ref{line:lineif}
	or the loop in line~\ref{line:loop2a2} executes at least once.
	Since $\BBB(a) \neq \emptyset$ every time the loop at line~\ref{line:loop2a2}
	is entered, an agent $a'$ is selected in line~\ref{line:rprimea2}.
	By the choice of $a'$, Lemma~\ref{lem:c1_unm_cond}, {\bf Fact 2}
	and the choice of $p$ in line~\ref{line:pickp},
	the dual update in line~\ref{line:za2} ensures that at least one edge $(a',p_j)$, for $p_j \mpref_{a'} M(a')$
	becomes matchable and $a'$ gets matched along that edge.
	By {\bf (P2)}, this modification does not demote $a'$ (if $a'$ was already matched).
	Therefore, either an unmatched agent (either $a$ in line~\ref{line:lineif}
	or $a'$ in line~\ref{line:pro1a2}) 
	is matched or at least one agent ($a'$ in line~\ref{line:pro1a2}) is promoted during an iteration.

	Thus after $O(m)$ iterations of the loop in line~\ref{line:loop_a},
	a fixed unmatched agent $a$ gets matched and the loop in line~\ref{line:loop_a} terminates.
	As mentioned earlier, the free-promotions routine takes $O(m^2)$ time.
	Thus, the loop in line~\ref{line:loop_a} terminates in $O(m^3)$ time for a fixed unmatched agent $a$
	and the loop in~\ref{line:loop1a2} terminates in $O(m^3$$\mid\hspace{-0.1cm}\AAA\hspace{-0.1cm}\mid)$ time. 
	By the termination condition of the loop, $M$ is an $\AAA$-perfect matching.
\end{proof}
	
\noindent{\bf Remark on the running time. }
We observe that the initial setting of dual variables takes 
$O(m$$\mid\hspace{-0.1cm}\AAA\hspace{-0.1cm}\mid)$ time because there are $O(m$$\mid\hspace{-0.1cm}\AAA\hspace{-0.1cm}\mid)$ valid triplets.
Since the algorithm guarantees {\bf (P2)}, with careful implementation of the free-promotions routine
and efficiently computing the threshold agents, the running time of the algorithm can be improved.
\qed

Finally, we show that the matching $M$ computed by Algorithm~\ref{algo:dualalgo_2}
is an $\ell_a$-approximation. 
\begin{lemma}\label{lem:costana}
Matching $M$ computed by Algorithm~\ref{algo:dualalgo_2} is an $\ell_a$-approximation of $\MINSUMSPC$.
\end{lemma}
\begin{proof}
	Let $\OPT$ be an optimal matching and $c(M)$ and $c(\OPT)$ denote the cost of $M$ and $\OPT$ respectively.
	By the LP duality, $c(\OPT) \geq \sum\limits_{a\in\AAA}{y_a}$.
	By {\bf (P4)}, $(a,p) \in M$ implies that the edge $(a,p)$ is tight.
	Thus, we have
\begin{align*}
c(M) = \sum\limits_{(a,p) \in M}{c(p)}
	&= \sum\limits_{(a,p) \in M} \Big(y_a + \sum_{\substack{p' = p\ \text{or}\\ p'\lpref_{a} p}}\ \ {\sum\limits_{a' \lpref_{p'} a}{z_{a,p',a'}}} - \sum\limits_{a \lpref_p a'}{z_{a',p,a}}\Big)\\
	&= \sum\limits_{a\in\AAA}{y_a} + \underbrace{\sum\limits_{(a,p) \in M} \Big(\sum_{\substack{p' = p\ \text{or}\\ p' \lpref_{a} p}}\ \ {\sum\limits_{a' \lpref_{p'} a}{z_{a,p',a'}}} - \sum\limits_{a\lpref_p a'}{z_{a',p,a}}\Big)}_{S(Z)}
\end{align*}
	where the first equality is from Eq.~\ref{eq:lp1}, the second equality is from Eq.~\ref{eq:lp5} and the third equality follows because $M$ is $\AAA$-perfect.
	Let $S(Z)$ denote the second summation in the above cost.
	Our goal is to show that $S(Z)$ is upper-bounded by $(\ell_a-1)\sum\limits_{a\in\AAA}{y_a}$ thereby implying that $c(M) \leq \ell_a\cdot\sum\limits_{a\in\AAA}{y_a}$.

We first note that all the $z$ variables are set to $0$ initially
	and they are updated only inside the loop at line~\ref{line:loop2a2}.
	We charge the update in every $z$ variable to a specific unmatched agent $a$ picked at line~\ref{line:picka} and upper-bound the total update in $z$ charged to $a$ in terms of $y_a$.
	Let $A'$ be the set of agents unmatched before the loop at line~\ref{line:loop1a2} is entered.
    During every iteration of the loop in line~\ref{line:loop1a2},
	an unmatched agent $a$ from $A'$ is picked and the loop in line~\ref{line:loop_a} executes until $a$ is matched.
	Suppose that after picking $a$ in line~\ref{line:picka}, the loop in line~\ref{line:loop_a} runs for $\kappa(a)$ iterations.
	Then, $y_a$ is incremented by $c_2-c_1$ for $\kappa(a)$ times and since $a$ is matched,
	it is not picked again at line~\ref{line:picka}.
	Thus, at the end of algorithm, $y_a = c_1 + \kappa(a) (c_2-c_1)$, that is $y_a \geq \kappa(a) (c_2-c_1)$.

	We first present a simpler analysis that proves an $(\ell_a+1)$-approximation.
	Recall that the $z$ variables are non-negative (Eq.~\ref{eq:lp7}).
	Thus, we upper-bound the total value of $z$ variables appearing in positive form in $S(Z)$.
	During the iterations $1$ to $\kappa(a)-1$, the algorithm must enter the {\bf else} part
	and in the $\kappa(a)\text{-}th$ iteration, the loop may or may not enter the {\bf else} part.
	Suppose the algorithm enters the {\bf else} part. Then by Lemma~\ref{lem:boundp},
	for a fixed $a$ when the algorithm enters the {\bf else} part,
	at most $\ell_a$ many $z$ variables are set to $c_2-c_1$.
	Thus, at most $\kappa(a) \ell_a (c_2-c_1)$ total update in $S(Z)$ occurs during execution
	of the loop in line~\ref{line:loop_a} when agent $a$ is picked.
	We charge this cost to agent $a$, thus agent $a \in A'$ is charged at most $\ell_a y_a$.
	Thus,
\begin{align*}
	c(M) = \sum\limits_{a\in\AAA}{y_a} + S(Z)
	&\leq \sum\limits_{a\in\AAA\setminus A'}{y_a} + \sum\limits_{a\in A'}{y_a} + \sum\limits_{a\in A'}{\ell_a y_a}\\
	&\leq (\ell_a+1) \sum\limits_{a\in\AAA}{y_a}
	\leq (\ell_a+1) c(\OPT)
\end{align*}

    Now, we proceed to a better analysis that shows an $\ell_a$-approximation.
	Recall that if $(a',p, a)$ is a valid triplet then the variable $z_{a',p, a}$ appears in the dual constraint
    of possibly multiple edges incident on $a'$ in positive form and in the dual constraint of exactly one edge,
	that is, the edge $(a,p)$ in negative form. We show that there exist certain valid triplets such that
	the corresponding $z$ variable occurring in positive form in the dual constraint of a matched edge 
	also appears in negative form in the dual constraint of another matched edge, thereby canceling out 
	their contribution in $S(Z)$. Thus, it is enough to upper-bound the update in $z$ variables that
	are {\em not} cancelled. We prove that the total update in such
    $z$ variables that is charged to an agent $a \in A'$ can be upper-bounded by $(\ell_a-1) y_a$
    instead of $\ell_a y_a$ as done earlier.

	Let $a \in A'$ be an arbitrary agent.
	Suppose that after $a$ is selected at line~\ref{line:picka},
	$a$ is matched to some program $\overline{p}$ and that $M(a) = p_k$ at the end of the algorithm.
	By {\bf (P2)}, $p_k = \overline{p}$ or $p_k \mpref_{a} \overline{p}$.
	Also, during iterations $1$ to $\kappa(a)-1$,
	$thresh(p_k) \neq a$ and the loop in line~\ref{line:loop2a2} executes.
	It implies that in each of the iterations, there exists an agent $a_j$
	such that $thresh(p_k) = a_j$ and $z_{a_j,p_k,a}$ is updated.
	Also, $a_j$ was matched to $p'$ such that $p' = p_k$ or $p' \mpref_{a_j} p_k$.
	By {\bf (P2)}, at the end of the algorithm, $M(a_j) = p'$ or $M(a_j) \mpref_{a'} p'$.
	Thus, the variable $z_{a_j,p_k,a}$ appears in positive form in the dual constraint of 
	the edge $(a_j,M(a_j))$.
	Since $(a,p_k) \in M$ and the variable $z_{a_j,p_k,a}$ appears in negative form in the 
	dual constraint of edge $(a,p_k)$.
	Therefore, the variable $z_{a_j,p_k,a}$ cancels out in $S(Z)$.
	This implies that for each of the iterations $1$ to $\kappa(a)-1$,
	at most $\ell_a-1$ many $z$ variables are set to $c_2-c_1$ such that they may not cancel out.
	We charge the update in these variables to $a$.

	In the last $\kappa(a)$-th iteration, $a$ gets matched.
	If $a$ is matched at line~\ref{line:lineif} then no $z$ variable is updated during this iteration.
	Otherwise, $a$ is matched in one of the iterations of the loop in line~\ref{line:loop2a2}
	by the free-promotions routine.
	Recall that by our assumption, $a$ is matched to $\overline{p}$ in this step.
	By the choice of $\overline{p}$ in the free-promotions routine,
	the edge $(a,\overline{p})$ must have been matchable, that is, it is tight and $thresh(\overline{p}) = a$.
	The fact that edge $(a,\overline{p})$ was tight implies (by {\bf Fact 2}) that no variable
	of the form $z_{\times,\overline{p},a}$ was updated so far inside the loop at line~\ref{line:loop2a2}
	during the $\kappa(a)$-th iteration.
	When $\BBB(a)$ is re-computed, $\overline{p} \notin \BBB(a)$ because $M(a) = \overline{p}$ at this step.
	Thus, in the subsequent iterations of the loop in line~\ref{line:loop2a2},
	no agent $a'$ could have selected $\overline{p}$ in line~\ref{line:pickp}.
	This implies that no $z$ variable of the form $z_{\times,\overline{p},a}$ is updated
	during the rest of the execution of the loop at line~\ref{line:loop2a2} of the $\kappa(a)$-th iteration.
	This implies that during the $\kappa(a)$-th iteration,
	the $z$ variables that are set to $c_2-c_1$ are of the form
	$z_{\times,\hat{p},a}$ where 
	$\overline{p} \neq \hat{p}$. By the fact that $\overline{p} \in \mathcal{N}(a)$, $\hat{p} \in \mathcal{N}(a)$
	and $\mid\hspace{-0.1cm}\mathcal{N}(a)\hspace{-0.1cm}\mid$ $\leq \ell_a$, the number such $z$ variables is at most $\ell_a-1$.

	Thus, during $\kappa(a)$ many iterations for the agent $a \in A'$ at most $\kappa(a) (\ell_a-1) (c_2-c_1)$ total update in $S(Z)$ is charged to $a$.
	Recall that $y_a \geq \kappa(a) (c_2-c_1)$.
	Thus, agent $a \in A'$ contributes at most $(\ell_a-1) y_a$ in $S(Z)$.
This gives
\begin{align*}
	c(M) = \sum\limits_{a\in\AAA}{y_a} + S(Z)
	&\leq \sum\limits_{a\in\AAA\setminus A'}{y_a} + \sum\limits_{a\in A'}{y_a} + \sum\limits_{a\in A'}{(\ell_a-1)\cdot y_a}\\
	&\leq \ell_a \sum\limits_{a\in\AAA}{y_a} \leq \ell_a \cdot c(\OPT)
\end{align*}
	This completes the proof of the lemma.
\end{proof}

This establishes Theorem~\ref{thm:minsumccq_2_cost_approximation}. 

%% file: 5-hardness.tex
\section{Hardness and inapproximability}\label{sec:hardness}

In this section, we prove hardness and inapproximability results for \MINSUMSPC. These results are extended versions of the hardness results in~\cite{iwoca}. Recall that under the \MINSUMSPC setting, envy-freeness and stability are equivalent.

\subsection{Constant factor inapproximability of \MINSUMSPC}\label{sec:ConstantFactorConstruction}

In this section, we show that \MINSUMSPC cannot be approximated within any constant factor. This result holds even when there are only two distinct costs in the instance and the preference lists follow master list ordering on both sides. 
We present a reduction from the Set cover problem.

\noindent {\bf The Set cover problem: } Given a universe $\mathcal{U}$ of $n$ elements $\{e_1, e_2, \ldots, e_n\}$ and a collection $\mathcal{C}$ of $m$ sets $\{C_1, C_2, \ldots, C_m\}$
where each set $C_j \subseteq \mathcal{U}$, a \textit{set cover} is a set $T\subseteq\mathcal{C}$ such that $\bigcup_{C\in T}C = \mathcal{U}$. Given an integer $k$, the goal is to decide whether there is a set cover with cardinality at most $k$.

A subset $C_j$ is said to cover an element $e_i$ if $e_i\in C_j$. A set $T\subseteq \mathcal{C}$ is said to cover element $e_i$ if there exists $C_j\in T$ such that $C_j$ covers $e_i$.

\noindent{\bf Reduction to \MINSUMSPC: }
Given a set cover instance $I=(\mathcal{U},\mathcal{C}, k)$, the corresponding \MINSUMSPC instance $I'$ is constructed as follows: for every element $e_i\in \mathcal{U}$, there is an element-agent $a_i$. For every subset $C_j\in\mathcal{C}$, there is a subset-program $c_j$ and $n$ dummy agents $u_j^l$ and $n$ dummy programs $w_j^l$ where $1 \le l \le n$. Therefore, in the reduced instance, there are $n+nm$ agents and $m+mn$ programs.

Next, we define the cost of the programs and the preference lists in the reduced instance.
For all subset-programs, $c(c_j) = 1$ and for all dummy programs, $c(w_j^l) = 0$. Note that in the reduced instance, there are only two distinct costs. For a set $Q$, let $\langle Q \rangle$ denote the elements of $Q$ ordered in a fixed but arbitrary way. Every element-agent $a_i$ ranks all subset-programs $c_j$ such that the corresponding element $e_i$ is in the subset $C_j$. Every dummy agent $u_j^i$ corresponding to the subset $C_j$ prefers the subset-program $c_j$ over the dummy program corresponding to $C_j$. Every subset-program $c_j$ prefers its corresponding dummy agents $u_j^l$ over the element-agents $a_i$ corresponding to elements in the subset $C_j$. Every dummy program ranks only its dummy agents. The preference lists are shown below. Note that $1 \le i \le n$, $1 \le j \le m$. \\
        
\begin{minipage}[c]{0.3\linewidth}
        {\bf Preference lists of agents}
    \begin{align*}
         a_i &: \langle\{c_j|e_i\in C_j\}\rangle \\
          u^l_j &: c_j \succ w^l_j 
    \end{align*}
\end{minipage} 
\begin{minipage}[c]{0.6\linewidth}
        {\bf Preference lists of programs}
    \begin{align*}
        \hspace{0.1cm} c_j (0, 1) &: \langle\{u_j^l| 1 \le l \leq n \}\rangle \succ \langle\{a_i|e_i\in C_j\} \\
        w^l_j (0, 0) &: u^l_j
    \end{align*}

\end{minipage}
        
\vspace{0.3cm}
    
\noindent We define $k'$ as $(k+1)n$. We also assume that $k > 1$, without loss of generality.

\begin{lemma}
\label{lem:setCoverToMatching}
Given a set cover $T$ in $I$ such that $|T| \le k$, we can construct an $\mathcal{A}$-perfect envy-free matching $M$ in $I'$ such that $c(M) \le k'$.
\end{lemma}

\begin{proof}
Let $C_{open} = \{c_j \mid C_j\in T\}$. Note that $|C_{open}| = |T|$. Construct $M$ by matching every element-agent $a_i$ to its highest-preferred program in $C_{open}$. Dummy agents corresponding to programs in $C_{open}$ are matched to their respective programs, and the remaining dummy agents are matched to their dummy programs. Thus the matching $M$ is $\mathcal{A}$-perfect.

We show that the constructed matching $M$ is envy-free. For element-agents $a_i$, any program $p$ such that $p \succ_{a_i} M(a_i)$ will have no agents matched to it (by construction). The same holds for dummy agents corresponding to programs not in $C_{open}$. The dummy agents corresponding to programs in $C_{open}$ are matched to their highest-preferred programs. Therefore, no agent envies another agent.

Next, we show that $c(M) \le k'$. For every program $c_j$ in $C_{open}$, some number of element-agents and all dummy agents of $c_j$ are matched to it in $M$. There are $n$ such dummy agents per program, while the element-agents together contribute a cost of $n$ over all programs. Therefore, $c(M) = |C_{open}|n+n = (|T|+1)n$. Given that $|T| \le k$, we get $c(M) \le (k+1)n = k'$.
\end{proof}
    
\begin{lemma}
\label{lem:CCQToSetCoverConstantApprox}
Given an $\mathcal{A}$-perfect envy-free matching $M$ in $I'$ such that $c(M)\leq \alpha k'$ for some constant $\alpha > 1$, we can construct a set cover $T$ in $I$ such that $|T|\leq 2\alpha k$.
\end{lemma}

\begin{proof}

Given $M$, define the set of programs $C_{open} = \{c_j \mid \exists a_i \text{ such that } M(a_i) = c_j\}$. Let $T = \{C_j \mid c_j\in C_{open}\}$. We show that the set $T$ is a set cover. 
Suppose for the contradiction that $T$ is not a set cover. Then there exists an element $e_i$ which is not covered by $T$. None of the potential partners of the corresponding element-agent $a_i$ are in $C_{open}$, so $a_i$ must be unmatched in $M$. This implies that $M$ is not $\mathcal{A}$-perfect, leading to a contradiction. Thus, $T$ is a set cover.

Next, we show that the size of $T$ is at most $2\alpha k$. Note that $|T| = |C_{open}|$. Since $M$ is an $\mathcal{A}$-perfect matching, each $a_i$ is matched in $M$, thereby contributing a cost of one unit each in $c(M)$. Moreover, since $M$ is envy-free, if an agent $a_i$ is matched to program $c_j$ then the $n$ dummy agents $u_j^l$ corresponding to $c_j$ must be matched to $c_j$ in $M$. They together contribute a cost of $n|C_{open}|$ since such $c_j \in C_{open}$. The dummy agents corresponding to each program $c_j \notin C_{open}$ may be matched to either $c_j$ or $w_j^l$. Thus, $c(M) \geq n+n|C_{open}| = n(1+|T|)$. Given that $c(M) \leq \alpha k'$, we get $n(|T|+1) \leq \alpha(k+1)n$.
This implies that $|T| \leq |T| + 1 \leq 2\alpha k$.

\end{proof}
         
Suppose, for the sake of contradiction, that there exists an $\alpha$-approximation algorithm for the \MINSUMSPC problem. For a given set cover instance, we can use the above reduction to construct a \MINSUMSPC instance. Then, using the $\alpha$-approximation algorithm and Lemma~\ref{lem:setCoverToMatching} and Lemma~\ref{lem:CCQToSetCoverConstantApprox}, we can get a $2\alpha$-approximation algorithm for the Set Cover problem. However, the Set Cover problem cannot be approximated within any constant factor unless $\Poly = \NP$~\cite{dinur2014}. This implies that for any constant $\alpha$, \MINSUMSPC does not admit an $\alpha$-approximation algorithm. 

Finally, we note that the following master list ordering over agents and programs holds in the reduced \MINSUMSPC instance.

\begin{align*}
    u_1^1\succ u_1^2\succ\dots\succ u_1^n\succ u_2^1\succ u_2^2\succ \dots \succ u_m^{n-1}\succ u_m^n \succ a_1 \succ a_2 \succ \dots \succ a_n \\
    c_1 \succ c_2 \succ \dots\succ c_m \succ w_1^1 \succ w_1^2 \succ \dots \succ w_1^n \succ w_2^1 \succ w_2^2 \succ \dots \succ w_m^{n-1} \succ w_m^n
\end{align*}

This establishes Theorem~\ref{thm:minsumsp_hardness}.

\noindent \textbf{Remark: }Note that in the \MINSUMSPC problem, stable matchings and envy-free matchings are the same. Thus, if we had an $\alpha$-approximation algorithm for the \MINSUMSP problem, we would have an $\alpha$-approximation algorithm for \MINSUMSPC, which would contradict Theorem \ref{thm:minsumsp_hardness}. Thus, the \MINSUMSP problem is also constant-factor inapproximable.
As mentioned earlier, this also follows from Theorem 2 in~\cite{chen}.

\subsection{$(\ell_a - \epsilon)$-inapproximability of \MINSUMSPC}

In this section, we show that our algorithmic result for the \MINSUMSPC problem  (Theorem~\ref{thm:minsumccq_2_cost_approximation}) is tight modulo Unique Games Conjecture~\cite{khot2008}. Specifically, we show that for any $\epsilon>0$, \MINSUMSPC does not admit an $(\ell_a-\epsilon)$ approximation algorithm.
We present a reduction from the Vertex cover problem, which is a special case of the Set cover problem defined in Section~\ref{sec:ConstantFactorConstruction}. The construction of the \MINSUMSPC instance presented below is similar to the construction presented in Section~\ref{sec:ConstantFactorConstruction}.

We first note the following: if \MINSUMSPC admits an approximation algorithm with guarantee $(\ell_a-\alpha)$ then it admits an approximation algorithm with guarantee $(\ell_a-\beta)$ for any constants $\beta<\alpha$. Therefore, it is enough to show that \MINSUMSPC does not admit an approximation algorithm with guarantee $(\ell_a-\epsilon)$ for $\epsilon \leq \frac{1}{2}$.

\bigskip

\noindent\textbf{The Vertex cover problem:} Given a graph $G(V, E)$ where $|V| = n$ and $|E| = m$, a set $T\subseteq V$ is called a vertex cover if for every edge $e\in E$, there is a vertex $v\in T$ such that $e$ is incident on $v$. Given an integer $k$, our goal is to decide whether there exists a vertex cover with cardinality at most $k$. An instance $(G,k)$ of the Vertex cover problem can be reduced to the instance of the Set cover problem by taking $\mathcal{U} = E$, $\mathcal{C} = \{\mathcal{N}(v) \mid v\in V\}$ and the same value of $k$.

\bigskip

\noindent{\bf Reduction to \MINSUMSPC: } Given a vertex cover instance $I$, construct the corresponding set cover instance $\overline{I}$. Then construct the \MINSUMSPC instance $I'$ from $\overline{I}$ as presented in Section~\ref{sec:ConstantFactorConstruction} with the following change: instead of constructing $n$ dummy agents and programs per subset $C_j\in \mathcal{C}$, construct $f(n,\epsilon)$-many such dummy agents and dummy programs corresponding to each subset $C_j$, where $f(n,\epsilon) = \frac{2n(1-\epsilon)}{\epsilon} \ge 1$ (since $\epsilon \leq \frac{1}{2}$).

Note that in the reduced instance, there are $n + f(n, \epsilon) m$ agents and $m + f(n, \epsilon) m$ programs.
The preference lists for all agents and programs are constructed identically as presented in Section~\ref{sec:ConstantFactorConstruction}. 
We define $k' = n + k f(n, \epsilon)$.

\begin{lemma}
    \label{lem:ellainapp2}
    Given a vertex cover $T$ in $I$ such that $|T| \le k$, we can construct an $\mathcal{A}$-perfect envy-free matching $M$ in $I'$ such that $c(M) \le k'$.
\end{lemma}

\begin{proof}
        Let $C_{open} = \{c_j \mid C_j\in T\}$. Note that $|C_{open}| = |T|$, as in Lemma~\ref{lem:setCoverToMatching}. 
        Again, $M$ is constructed by matching every element-agent $a_i$ to its highest-preferred program in $C_{open}$. Dummy agents corresponding to programs in $C_{open}$ are matched to their respective programs, and the remaining dummy agents are matched to their dummy programs. As in Lemma~\ref{lem:setCoverToMatching}, we see that $M$ is $\mathcal{A}$-perfect and envy-free. Moreover, $c(M) \le k'$.
\end{proof}

\begin{lemma}\label{lem:ellainapp1}
    If $I'$ admits an $\mathcal{A}$-perfect envy-free matching $M$ with $c(M) \leq \left (\ell_a - \epsilon \right)k'$ for $\frac 12 \ge \epsilon > 0$, then $I$ admits a vertex cover $T$with $|T| \leq \left(2-\epsilon' \right)k$ where $\epsilon' = \frac{\epsilon}{2}$.
\end{lemma}

\begin{proof}

Given an $\mathcal{A}$-perfect envy-free matching $M$, define the set of opened programs $C_{open}$ and construct a vertex cover $T$ using $C_{open}$, as done in the proof of Lemma~\ref{lem:CCQToSetCoverConstantApprox}. Hence, $|T| = |C_{open}|$. As it follows from Lemma~\ref{lem:CCQToSetCoverConstantApprox}, $T$ is a valid set cover. Next, we show that $|T|\leq \left(2-\frac{\epsilon}{2} \right)k$. Since we have $f(n, \epsilon)$-many dummy agents for every program, specifically for every $c_j \in C_{open}$, the cost $c(M) \geq n + |T| \cdot f(n, \epsilon)$. Also, in the reduced instance $I'$, we have $\ell_a=2$ since each element of $\overline{I}$ is contained in exactly two sets (corresponding to the two end-points of that edge in $I$), so the element-agents have only two programs in their preference lists, dummy agents always have only two programs in their preference list. Therefore, we get

\begin{align*}
    n +|T| \cdot f(n, \epsilon) & \leq c(M) \le (\ell_a-\epsilon)k'
             = (\ell_a - \epsilon)(n + kf(n, \epsilon))\\
    \implies |T| &\leq (\ell_a - \epsilon)k + \frac{(\ell_a-\epsilon-1)n}{f(n, \epsilon)}
    = (2-\epsilon)k + \frac{(1-\epsilon)n}{f(n, \epsilon)} & (\ell_a = 2)
\end{align*}

Substituting $f(n, \epsilon) = \frac{n(1 - \epsilon)}{\left(\frac{\epsilon}{2}\right)}$ followed by $k \geq 1$, we get

\begin{align*}
    |T| &\leq (2-\epsilon)k + \frac{\epsilon}{2}
    \leq (2-\epsilon)k + \frac{k\epsilon}{2}
    =\left(2-\frac{\epsilon}{2} \right)k = (2 - \epsilon') k
\end{align*}

where $\epsilon' \equiv \frac{\epsilon}2 > 0$.
\end{proof}

If the \MINSUMSPC problem admits an $(\ell_a-\epsilon)$-approximation algorithm ($\epsilon > 0$), then using Lemma~\ref{lem:ellainapp1} and Lemma~\ref{lem:ellainapp2} we get a $(2-\epsilon')$-approximation algorithm ($\epsilon' > 0$) for the Vertex cover problem (where $\epsilon' = \frac{\epsilon}{2}$). However, under the Unique Games Conjecture, it is known that the Vertex cover problem cannot be approximated within a factor of $(2-\epsilon')$ for $\epsilon' > 0$~\cite{khot2008}. Therefore, the \MINSUMSPC problem does not admit an $(\ell_a-\epsilon)$-approximation algorithm, for any $\epsilon>0$.

We further notice that in the reduced instance, there is a master list ordering over agents and programs, as described in Section \ref{sec:ConstantFactorConstruction}.
This establishes Theorem~\ref{thm:minsumsp_inapproximability}.

\noindent \textbf{Remark:} An argument similar to that in Section~\ref{sec:ConstantFactorConstruction} shows that Theorem~\ref{thm:minsumsp_inapproximability} implies $(\ell_a-\epsilon)$-inapproximability for the \MINSUMSP problem as well.

%% file: 6-concl.tex
\section{Concluding remarks}\label{sec:concl}
In this work we propose and investigate the generalized capacity planning problem
for the many-to-one matchings under two-sided preferences.
Motivated by the need to match every agent,
we propose a setting wherein costs control the extent to
which a program is matched. We aim to compute a stable matching in an
optimally cost-augmented instance such that it matches every agent.
We investigate two optimization problems.
We prove that the $\MINMAXSP$ problem is efficiently solvable
but the $\MINSUMSP$ problem turns out to be $\NP$-hard.
We present approximation algorithms for $\MINSUMSP$ with varying approximation guarantees
and an improved approximation algorithm for a special hard case.
A specific open direction is to bridge the gap between the upper bound and lower bound for 
general instances of the $\MINSUMSP$ problem.
It is also interesting to extend the LP algorithm for general instances.

%% file: appendix.tex
\appendix

\section{Appendix}

\subsection{Challenges in designing a primal-dual algorithm for \MINSUMSPC}

A standard primal-dual approach for the $\MINSUMSPC$ problem would be to begin with
 a dual feasible solution. The algorithm then repeatedly updates the dual till we obtain a primal feasible solution using the tight edges w.r.t. to
the dual setting.
We illustrate the challenges in using such an approach for the general $\MINSUMSPC$ problem. 
Consider the $\MINSUMSPC$ instance in Fig.~\ref{fig:challenges_1}.
Recall that the tuple $(q,c)$ preceding a program indicates that the initial quota and cost of that program
is $q$ and $c$ respectively. Since the instance in Fig.~\ref{fig:challenges_1}
is of the \MINSUMSPC problem, the initial quotas are 0 for each program.

\input{appendix-challenges_fig}

Assume that we begin with an initial dual setting where all dual variables are set to $0$. 
The matching $M  = \{(a_1, p_0), (a_2, p_0), (a_3, p_0), (a_4,p_0)\}$ obtained on the tight edges
is envy-free but does not match agent $a_5$ and hence is not primal feasible.
Since no edge incident on $a_5$ is tight (slack on $(a_5, p_2)$ and $(a_5, p_3)$ is 6 and 11 respectively) we 
can set $y_5$ to 6 while maintaining dual feasibility. We observe that while this update makes the edge $(a_5, p_2)$
tight, adding the edge to the matching $M$ introduces an envy pair -- namely $a_4$ envying $a_5$. 
We note that this is our first difficulty,
that is, while there are tight edges incident on an unmatched
agent, none of them may be matchable. 

The second difficulty stems from the following: in order to match $a_5$ along the (non-matchable) tight edge $(a_5,p_2)$ we must first resolve the potential envy pair $(a_4, a_5)$, that is, we must {\em promote} agent $a_4$. With the current dual setting, $y_4$ cannot be increased hence a natural way is to
update a $z$ variable. This can indeed be achieved by setting $z_{{a_4}, {p_2}, {a_5}} = 1$, thus making $(a_4, p_1)$ tight. However, as encountered earlier, this edge is not matchable, since matching $a_4$ to $p_1$ introduces several other envy pairs. Note that this chain of potential envy resolutions is triggered by the unmatched agent $a_5$. Since, this chain can be arbitrarily long, several $z$ updates may be required.
It is not immediate if these updates in $z$ variables can be charged to an update in some $y$ variable, thereby achieving a reasonable approximation ratio.

However, as seen in Section~\ref{sec:two_cost_algo}, for the restricted hard case of two distinct costs, we are able to resolve these challenges.

\subsection{Example illustrating the execution of Algorithm~\ref{algo:dualalgo_2}}

\begin{figure}[!h]
\begin{center}
\begin{minipage}{0.2\textwidth}
\begin{align*}
      a_1 &: p_1\mpref p_2\mpref p_0\\
      a_2 &: p_2\mpref p_3\mpref p_0\\
      a_3 &: p_1\mpref p_2\mpref p_3\\
  \\[-10pt]
  \hline
  \\[-10pt]
              (0,0)\ p_0 &: a_1\mpref a_2\\
              (0,1)\ p_1 &: a_1\mpref a_3\\
              (0,1)\ p_2 &: a_1\mpref a_2\mpref a_3\\
              (0,1)\ p_3 &: a_2\mpref a_3
              \end{align*}
  \end{minipage}\hfill
      \begin{minipage}{0.65\textwidth}
      \vspace{0.5cm}
    \begin{itemize}
        \item $M = \{(a_1,p_0), (a_2,p_0)\}$
        \item \textcolor{blue}{(\ref{step:atight})} $a = a_3$, $y_{a_3} = 1$, tight edges on $a_3$ are $\{(a_3,p_1)$, $(a_3,p_2)$, $(a_3,p_3)\}$, $thresh(p_1) = thresh(p_2) = a_1$ and $thresh(p_3) = a_2$
        \item \textcolor{blue}{(\ref{step:ba})} $\BBB(a_3) = \{p_1, p_2, p_3\}$
        \item \textcolor{blue}{(\ref{step:pickp},\ref{step:zupdt})} let $a' = a_1$, then $p = p_2$, $z_{a_1,p_2,a_3} = 1$
        \item \textcolor{blue}{(\ref{step:matchaprime})} Tight edges on $a_1$ are $\{(a_1, p_1), (a_1, p_2)\}$, $p' = p_1$, $M = \{(a_1,p_1), (a_2,p_0)\}$, tight edges on $a_3$ are $\{(a_3,p_1), (a_3,p_3)\}$
        \item \textcolor{blue}{(\ref{step:fp})} $thresh(p_1) = a_3$, $M = \{(a_1,p_1), (a_2,p_0), (a_3,p_1)\}$
        \item \textcolor{blue}{(\ref{step:recompBa})} $\BBB(a_3) = \emptyset$
    \end{itemize}
              \end{minipage}
        \vspace{0.3cm}
        \caption{A $\CCQCC$ instance.
        An execution of Algorithm~\ref{algo:dualalgo_2} is illustrated by giving the state of the algorithm. The \textcolor{blue}{blue} content in the brackets correspond to the labels of steps mentioned in the description of the algorithm.}
    \label{fig:exinst}
    \end{center}
\end{figure}

%% file: appendix-challenges_fig.tex
\begin{figure}[!ht]
\begin{minipage}{0.4\textwidth}
	\begin{align*}
		\s_1 &: \cc_1 \mpref \cc_0\\
		\s_2 &: \cc_1 \mpref \cc_0\\
		\s_3 &: \cc_1 \mpref \cc_0\\
		\s_4 &: \cc_1 \mpref \cc_2 \mpref \cc_0\\
		\s_5 &: \cc_2 \mpref \cc_3
	\end{align*}
\end{minipage}\hfill
	\begin{minipage}{0.4\textwidth}
		\begin{align*}
			(0,0)\ \cc_0 &: \s_1 \mpref \s_2 \mpref \s_3 \mpref \s_4\\
			(0,1)\ \cc_1 &: \s_1 \mpref \s_2 \mpref \s_3 \mpref \s_4\\
			(0,6)\ \cc_2 &: \s_4 \mpref \s_5\\
			(0,11)\ \cc_3 &: \s_5
		\end{align*}
	\end{minipage}\hfill

	\caption{A $\MINSUMSPC$ instance used in illustrating the challenges}
	\label{fig:challenges_1}
\end{figure}